%% file: paper.tex
\documentclass[runningheads,a4paper]{llncs}

\usepackage{a4wide}
\usepackage[latin1]{inputenc}
\usepackage{graphicx}
\usepackage{amsfonts}
\usepackage{amssymb}
\usepackage{amsmath}
\usepackage{framed}
\usepackage{latexsym}
\usepackage{breakcites}
\usepackage[table]{xcolor}
\usepackage{color}
\usepackage{url}
\usepackage{tikz}
\usetikzlibrary{shapes}
\usetikzlibrary{patterns}
\usetikzlibrary{decorations.pathreplacing}
\usepackage{algorithm}
\usepackage[noend]{algpseudocode}
\usepackage{enumerate}
\usepackage{stmaryrd}
\usepackage{thm-restate}
\usepackage{booktabs}
\input{macros.tex}
\newcommand{\COMMENT}[1]{}

\begin{document}
	
\title{Ternary Syndrome Decoding with Large Weight}
\author{R\'emi Bricout \inst{1,2} \and Andr\'e Chailloux \inst{2} \and Thomas Debris-Alazard\inst{1,2} \and Matthieu Lequesne \inst{1,2}}
\institute{Sorbonne Universit\'{e}s, UPMC Univ Paris 06 \and Inria, Paris\\
	\email{\{remi.bricout, andre.chailloux, thomas.debris, matthieu.lequesne\}@inria.fr}}
\maketitle

\begin{abstract}
\input{abstract.tex}
\end{abstract} 


\section{Introduction}
\input{1-introduction.tex}

\section{A General Framework for Solving the Syndrome Decoding Problem} 
\label{sec:genFrameSD}
\input{2-general-framework.tex}

\section{Ternary Subset Sum with the Generalized Birthday Algorithm}
\label{sec:GBA}
\input{3-GBA.tex}

\section{Ternary Subset Sum Using Representations}
\label{sec:Reps}
\input{4-representations.tex}

\section{New Parameters for the WAVE Signature Scheme}
\label{sec:Wave}
\input{5-Wave.tex}

\section{Hardest Instances of Ternary Syndrome Decoding}
\label{sec:Hard}
\input{6-hardest-instances.tex}

\section{Conclusion}
\input{7-conclusion.tex}

\newpage
\bibliographystyle{alpha}
\bibliography{codecrypto}

\newpage
\appendix
\section{Appendix: Ternary Representations}
\label{sec:ternaryReps}
\input{appendix.tex}

\end{document}

%% file: macros.tex
\newcommand\numberthis{\addtocounter{equation}{1}\tag{\theequation}}

\makeatletter
\newcounter{restate}

\makeatother

\newcommand{\cS}{\mathcal{S}}

\newcommand{\OOt}[1]{\tilde{O}\left( #1 \right)}

\newcommand{\bv}{{\mathbf{b}}}

\newcommand{\ev}{{\mathbf{e}}}

\newcommand{\mv}{{\mathbf{m}}}

\newcommand{\sv}{{\mathbf{s}}}

\newcommand{\xv}{{\mathbf{x}}}
\newcommand{\yv}{{\mathbf{y}}}

\newcommand{\Hm}{{\mathbf{H}}}

\newcommand{\Sm}{{\mathbf{S}}}

\newcommand{\mat}[1]{\ensuremath{\boldsymbol{#1}}}
\newcommand{\un}{{\mat{1}}}

\newcommand{\Ic}{{\mathcal I}}

\newcommand{\Lc}{{\mathcal L}}

\newcommand{\SD}{\mathrm{SD}}

\renewcommand{\SS}{\mathrm{SS}}
\newcommand{\SSNZC}{\mathrm{SSNZC}}

\newtheorem{notation}{Notation}

\newcommand{\F}{\mathbb{F}}

\makeatletter
\newcommand*{\transp}{{\mathpalette\@transpose{}}}
\newcommand*{\@transpose}[2]{\raisebox{\depth}{$\m@th#1\intercal$}}
\makeatother
\newcommand{\transpose}[1]{{#1}^{\transp}}

\newcommand{\eqdef}{\mathop{=}\limits^{\triangle}}

\newcommand{\DOOM}{\ensuremath{\mathrm{DOOM}}}

\newcommand{\wt}[1]{|#1|}

\newcommand{\zo}{\{0,1\}}

\newcommand{\wave}{Wave}

\newcommand{\WGV}{W_{\textup{GV}}}
\newcommand{\Hmpi}{{\mathbf{H}_{\pi}}}
\newcommand{\Ppl}{{\mathcal{P}_{p,\ell}}}
\newcommand{\PGESS}{{\text{PGE+SS} }}

\newcommand\bin[2]{\begin{pmatrix}#1\\#2\end{pmatrix}}
\DeclareMathOperator{\abs}{abs}
\DeclareMathOperator{\Nrep}{Nrep}

%% file: abstract.tex
The Syndrome Decoding problem is at the core of many code-based
cryptosystems.  In this paper, we study ternary Syndrome Decoding in
large weight. This problem has been introduced in the Wave signature
scheme but has never been thoroughly studied. We perform an
algorithmic study of this problem which results in an update of the
Wave parameters. On a more fundamental level, we show that ternary
Syndrome Decoding with large weight is a really harder problem than
the binary Syndrome Decoding problem, which could have several
applications for the design of code-based cryptosystems.

%% file: 1-introduction.tex
Syndrome decoding is one of the oldest problems used in coding theory
and cryptography \cite{M78}. It is known to be NP-complete
\cite{BMT78} and its average case variant is still believed to be hard
forty years after it was proposed, even against quantum
computers. This makes code-based cryptography a credible candidate for
post-quantum cryptography. There has been numerous proposals of
post-quantum cryptosystems based on the hardness of the Syndrome
Decoding (SD) problem, some of which were proposed for the NIST
standardization process for quantum-resistant cryptographic
schemes. Most of them are qualified for the second round of the
competition \cite{BIKE,ACPTT17,HQC17,BBCPS19,BCLMNPPSSSW17}.  It is
therefore a significant task to understand the computational hardness
of the Syndrome Decoding problem.

Informally, the Syndrome Decoding problem is stated as follows. Given
a matrix $\Hm\in\F_{q}^{(n-k)\times n}$, a vector $\sv\in\F_{q}^{n-k}$
and a weight $w \in \llbracket 0,n \rrbracket$, the goal is to find a
vector $\ev\in\F_{q}^n$ such that $\Hm\transpose{\ev}=\transpose{\sv}$
and $\wt{\ev}=w$, where $\wt{\ev}$ denotes the Hamming weight, namely
$\wt{\ev} = |\{i : \ev_i \neq 0\}|$. The binary case, \textit{i.e.}
when $q = 2$ has been extensively studied. Even before this problem
was used in cryptography, Prange \cite{P62} constructed a clever
algorithm for solving the binary problem using a method now referred
to as Information Set Decoding (ISD).

\subsection{Binary vs. Ternary Case}

The binary case of the SD problem has been thoroughly studied. Its
complexity is always studied for relative weight $W := \frac{w}{n} \in
[0,0.5]$ because the case $W > 0.5$ is equivalent (see Remark
$\ref{rem:symetry}$ in Section \ref{sec:genFrameSD}). However, this
argument is no longer valid in the general case ($q \geq 3$). Indeed,
the large weight case does not behave similarly to the small weight
case, as we can see on Figure \ref{fig:binary-ternary}.

\begin{figure} \centering
  \includegraphics[width=.8\textwidth]{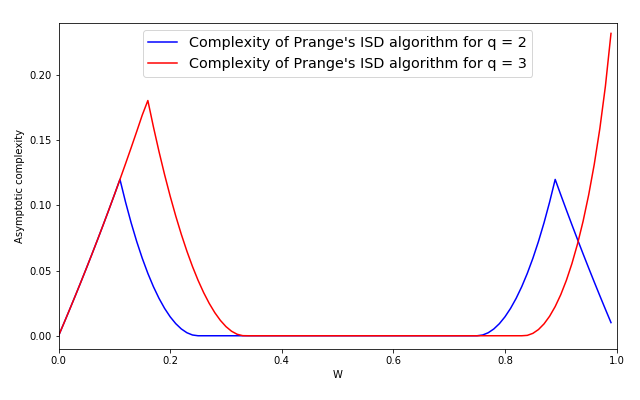}
  \caption{Asymptotic complexity of Prange's ISD algorithm for $R :=
\frac{k}{n} = 0.5$.}
  \label{fig:binary-ternary}
\end{figure}

The general case $q \geq 3$ has received much less attention than the
binary case. One possible explanation for this is that there were no
cryptographic applications for the general case. This has recently
changed. Indeed, a new signature scheme named \wave\ was recently
proposed in \cite{DST18}, based on the difficulty of SD on a ternary
alphabet and with large weight. This scheme makes uses of the new
regime of large weight induced by the asymmetry of the ternary
case. Therefore, in addition to the algorithmic interest of studying
the general syndrome decoding problem, the results of this study can be
applied to a real cryptosystem.

Another reason why the general case $q \geq 3$ has been less studied
is that the Hamming weight measure becomes less meaningful as $q$
grows larger.  Indeed, the Hamming weight only counts the number of
non-zero elements but not their repartition. Hence, the weight loses a
significant amount of information for large values of $q$. Therefore,
$q = 3$ seems to be the best candidate to understand the structure of
the non-binary case without losing too much information.

\subsection{State of the Art for $q \geq 3$}

Still, there exist some interesting results concerning the SD problem
in the general $q$-ary case. Coffey and Goodman \cite{CG90} were the
firsts to propose a generalization of Prange's ISD algorithm to
$\F_q$. Following this seminal work, most existing ISD algorithms were
extended to cover the $q$-ary case.  In 2010, Peters \cite{P10}
generalized Stern's algorithm. In his dissertation thesis, Meurer
\cite{M12} generalized the BJMM algorithm. Hirose \cite{H16} proposed
a generalization of Stern's algorithm with May-Ozerov's approach
(using nearest neighbors) and showed that for $q \geq 3$ this does not
improve the complexity compared to Stern's classical approach. Later,
Gueye, Klamti and Hirose \cite{GKH17} extended the BJMM algorithm with
May-Ozerov's approach and improved the complexity of the general SD
problem.  A result from Canto-Torres \cite{C17a} proves that all
ISD-based algorithms converge to the same asymptotic complexity when
$q \to \infty$.  Finally, a recent work \cite{IKRRW18} proposed a
generalization of the ball-collision decoding over $\F_q$.

All these papers focus solely on the SD problem for relative weight $W
< 0.5$.  None of them mentions the case of large weight. The claimed
worst case complexities in these papers should be understood as the
worst case complexity for the SD problem with relative weight $W <
0.5$, but as we can see on Figure \ref{fig:binary-ternary} the highest
complexity is actually reached for large relative weight.

\subsection{Our Contributions}

Our contribution consists in a general study of ternary syndrome
decoding with large weight. We first focus on the Wave signature
scheme \cite{DST18} and present the best known algorithmic attack on
this scheme. We then look more generally at the hardest instances of
the ternary syndrome decoding with large weight and show that this
problem seems significantly harder than the binary variant, making it
a potentially very interesting problem for code-base cryptography.

\paragraph{The PGE+SS framework.} A first minor contribution consists
in a modular description of most ISD-based algorithms. All these
algorithms contain two steps. First, performing a partial Gaussian
elimination (PGE), and then, solving a variant of the Subset Sum
problem (SS). This was already implicitly used in previous papers but
we want to make it explicit to simplify the analysis and hopefully
make those algorithms easier to understand for non-specialists.

\paragraph{Ternary SD with large weight.} We then study specifically
the SD problem in the ternary case and for large weights. From our
modular description, we can focus only on finding many solutions of a
specific instance of the Subset Sum problem. At a high level, we
combine Wagner's algorithm \cite{W02} and representation techniques
\cite{BCJ11,BJMM12} to obtain our algorithm. Our first takeaway is
that, while representations are very useful to obtain a unique
solution (as in \cite{BCJ11}), there are some drawbacks in using them
to obtain many solutions. These drawbacks are strongly mitigated in
the binary case as in \cite{BJMM12} but it becomes much harder for
larger values of $q$. We manage to partially compensate this by
changing the moduli size, the place and the number of representations.
For instance, for the \wave\ \cite{DST18} parameters, we derive an
algorithm that is a Wagner tree with seven floors where the last two
floors have partial representations and the others have none.

\paragraph{New parameters for \wave.}We then use our algorithms to
study the complexity of the \wave\ signature scheme, for which we
significantly improve the original analysis. We show that the key
sizes of the original scheme presented for 128 bits of security have
to be more than doubled, going roughly from 1Mb to 2.2Mb, to achieve
the claimed security. This requires to study the Decode One Out of
Many (DOOM) problem, on which {\wave} actually relies. This problem
corresponds to a multiple target SD problem. More precisely, given $N$
syndromes $(\sv_1,\ldots,\sv_N)$ ($N$ can be large, for example $N =
2^{64}$) the goal is to find an error $\ev$ of Hamming weight $w$ and
an integer $i$ such that $\ev\transpose{\Hm} = \sv_i$. 

\paragraph{Hardest instances of the ternary SD with large weight.}

Next, we look at the hardest instances of the ternary SD with large
weight problem. We study the standard ISD algorithms and show that for
all of them, the hardest instances occur for $R \approx 0.369$ and
$W = 1$ (still in the case $q = 3$). Unsurprisingly, for equivalent
code length and dimension, ternary syndrome decoding is harder than
its binary counterpart. But this is due to the fact that the input
matrix contains more information, since its elements are in $\F_3$,
hence the input size is $\log_2(3)$ times larger than a binary matrix
with equivalent dimensions.

A more surprising conclusion of our work is that ternary syndrome
decoding is significantly harder than the binary case \emph{for
  equivalent input size}, that is, when normalizing the exponent by a
factor $\log_2(q)$. This new result is in sharp contrast with all the
previous work on $q$-ary syndrome decoding that showed that the
problem becomes simpler as $q$ increases. This is due to the fact that
all the previous literature only considered the small weight case
while we now take large weights into account.

Table \ref{table:size1} represents the minimum input size for which the
underlying syndrome decoding problem offers $128$ bits of security,
\textit{i.e.} the associated algorithm needs at least $2^{128}$
operations to solve the problem.

\begin{table}[h!]  \centering
	\begin{tabular}{c c c} \toprule \quad Algorithm \quad & \quad
$q=2$ \quad & \quad $q = 3$ and $W > 0.5$ \quad \\ \midrule Prange &
275 & 44 \\ Dumer/Wagner & 295 & 83 \\ BJMM/Our algorithm & 374 & 99
\\ \bottomrule
        \end{tabular} \vspace*{+0.3cm}
	\caption{Minimum input sizes (in kbits) for a time complexity
of $2^{128}$. \label{table:size1}}
\end{table}

We want to stress again that those input sizes in the ternary case
take into account the fact that the matrix elements are in $\F_3$. So
the increase in efficiency is quite significant and the ternary SD
could efficiently replace its binary counterpart when looking for a
hard code-based problem.

\subsection{Notations}
\label{sec:nota} 

We define here some notations that will be used throughout the
paper. The notation $x \eqdef y$ means that $x$ is defined to be equal
to $y$.  We denote by $\mathbb{F}_{q} = \{0,1,\cdots,q-1\}$ the finite
field of size $q$. Vectors will be written with bold letters (such as
$\ev$) and uppercase bold letters will be used to denote matrices
(such as $\Hm$). Vectors are in row notation.  Let $\xv$ and $\yv$ be
two vectors, we will write $(\xv,\yv)$ to denote their concatenation. Finally, we denote by $\llbracket a,b \rrbracket$ the set $\{\tilde{a},\tilde{a}+1,\dots,\tilde{b}\}$ where $\tilde{a} = \lfloor a \rfloor$ 
and $\tilde{b} = \lfloor b \rfloor$.

%% file: 2-general-framework.tex
\subsection{The Syndrome Decoding Problem}

The goal of this paper is to study the Syndrome Decoding problem,
which is at the core of most code-based cryptosystems.

\begin{restatable}{problem}
{problemSD}[Syndrome Decoding - $\SD(q,R,W)$]
\label{prob:SDH}

\begin{tabular}{ll} 
  Instance: &
              \quad $\Hm\in\F_{q}^{(n-k)\times n}$ of full rank,\\ 
              & \quad $\sv\in\F_{q}^{n-k}$
              (usually called the \emph{syndrome}). \\ 
  Output: & \quad $\ev\in\F_{q}^n$ such
            that $\wt{\ev}=w$ and $\ev\transpose{\Hm} =\sv$,
\end{tabular} $ \ $
 
where $k \eqdef \lceil Rn\rceil$, $w \eqdef \lceil W n\rceil$ 
and $\wt{\ev} \eqdef |\{i : \ev_i \neq0\}|$.
\end{restatable}

The problem $\SD(q,R,W)$ is parametrized by the field size $q$, the
rate $R \in [0,1]$ and the relative weight $W \in [0,1]$. We are
always interested in the average case complexity (as a function of
$n$) of this problem, where $\Hm$ is chosen uniformly at random and
$\sv$ is chosen uniformly from the set
$\{\ev\transpose{\Hm} : |\ev| = w \}$. This ensures the existence of a
solution for each input and corresponds to the typical situation in
cryptanalysis. More generally, the following proposition gives the
average expected number of solutions

\begin{proposition}\label{propo:GVdist} Let $n,k,w$ be integers with
	$k \leq n$ and $\sv \in \F_q^{n-k}$. The expected number of
	solutions of $\ev\transpose{\Hm} = {\sv}$ in $\ev$ of weight $w$
	when $\Hm$ is chosen uniformly at random in $\F_q^{(n-k)\times n}$
	is given by:
	\begin{displaymath}
	\frac{\binom{n}{w}(q-1)^{w}}{q^{n-k}}.
	\end{displaymath}
\end{proposition}

\begin{proof} 
	This is simple combinatorics. The numerator
	corresponds to the number of vectors $\ev'$ of weight $w$. The
	denominator corresponds to the inverse of the probability over $\Hm$ that
	$\ev\transpose{\Hm} = \transpose{\sv}$ for $\ev \neq \mathbf{0}$. \qed
\end{proof}

\begin{remark}
  The matrix length $n$ is not considered as a parameter of the
  problem since we are only interested in the asymptopic complexity,
  that is the coefficient $F(q,R,W)$ (which does not depend on $n$)
  such that the complexity of the Syndrome Decoding problem for a
  matrix of size $n$ can be expressed as $2^{n(F(q,R,W)+o(1))}$.
\end{remark}

\subsubsection{State of the Art on $\F_2$.}

This problem was mostly studied in the case $q=2$. Depending on the
parameters $R$ and $W$, the complexity of the problem can greatly
vary. Let us fix a value $R$, and let $\WGV$ denote the
Gilbert-Varshamov bound, that is $\WGV \eqdef h_2^{-1}(1-R)$ where
$h_2$ is the binary entropy function restricted to the input space
$\left[0,\frac{1}{2}\right]$. For $W \in \left[0,\frac{1}{2}\right]$,
there exist three different regimes.

\begin{enumerate}
\item $W \approx \WGV$. When $W$ is close to $\WGV$, there is on
  average a small number of solutions. This is the regime where the
  problem is the hardest and where it is the most studied. To the best
  of our knowledge, we only know two code-based cryptosystems in this
  regime, namely the CFS signature scheme \cite{CFS01} and the
  authentication scheme of Stern \cite{S93}.

\item $W \gg \WGV$. In this case, there are on average exponentially
  many solutions and this makes the problem simpler. When $W$ reaches
  $\frac{1-R}{2}$, the problem can be solved in average polynomial
  time using Prange's algorithm \cite{P62}. There is a cryptographic
  motivation to consider $W$ much larger than $\WGV$, for instance to
  build signatures schemes following the \cite{GPV08} paradigm as it
  was done in \cite{DST17a} but one has to be careful to not make
  $\SD$ too simple.
  
\item $W \ll \WGV$. In this regime, we have with high probability a
  unique solution. However, the search space, \textit{i.e.} the set of
  vectors $\ev$ st. $\wt{\ev} = \lceil W n\rceil$ is much smaller than
  in the other regimes. The original McEliece system \cite{M78} or the
  QC-MDPC systems \cite{MTSB12} are in this regime.
\end{enumerate}

\begin{remark}
\label{rem:symetry}
Solving $\SD(2,R,W)$ for $W \in \left[\frac{1}{2},1\right]$ and the
instance $(\Hm,\sv)$ can be reduced to one of the above-mentioned
cases using $\SD(2,R,1-W)$ and the instance
$(\Hm,\sv + \mathbf{1}\transpose{\Hm})$ where $\mathbf{1}$ denotes the
vector with all its components equal to $1$.
\end{remark}

\begin{remark}
\label{rem:asymetry}
Contrary to the binary case, when $q \geq 3$ the case of large
relative weight can not be reduced to that of small relative weight
using the trick of Remark \ref{rem:symetry}. In fact, the problem has a
quite different behavior in small and large weights, see Figure
\ref{fig:binary-ternary}.
\end{remark}

\subsection{The \PGESS Framework in $\F_q$} 
\label{subsec:PGESS} 

The $\SD$ problem has been extensively studied in the binary
case. Most algorithms designed to solve this problem
\cite{D91,MMT11,BJMM12} follow the same framework:
\begin{enumerate}
\item perform a partial Gaussian elimination (PGE) ;
\item solve the Subset Sum problem (SS) on a reduced instance. 
\end{enumerate}

We will see how we can extend this framework to the non-binary case.
Our goal here is to describe the \PGESS framework for solving
$\SD(q,R,W)$.  Fix $\Hm\in\F_{q}^{(n-k)\times n}$ of full rank and
$\sv\in\F_{q}^{n-k}$. Recall that we want to find $\ev\in\F_{q}^n$
such that $\wt{\ev}= w \eqdef\ \lceil W n \rceil$ and
$\Hm\transpose{\ev}=\transpose{\sv}$. Let us introduce $\ell$ and $p$,
two parameters of the system, that we will consider fixed for now. In
this framework, an algorithm for solving $\SD(q,R,W)$ will consist of
$4$ steps: a permutation step, a partial Gaussian Elimination step, a
Subset Sum step and a test step.

\begin{enumerate}
\item \emph{Permutation step.} Pick a random permutation $\pi$. Let
  $\Hmpi$ be the matrix $\Hm$ where the columns have been permuted
  according to $\pi$. We now want to solve the problem $\SD(q,R,W)$ on
  inputs $\Hmpi$ and $\sv$.

\item \emph{Partial Gaussian Elimination step.} If the top left square
  submatrix of $\Hmpi$ of size $n-k-\ell$ is not of full rank, go back
  to step $1$ and choose another random permutation $\pi$. This
  happens with constant probability. Else, if this submatrix is of
  full rank, perform a Gaussian elimination on the rows of $\Hmpi$
  using the first $n-k-\ell$ columns. Let
  $\Sm \in \F_q^{(n-k)\times (n-k)}$ be the invertible matrix
  corresponding to this operation. We now have two matrices
  $\Hm' \in \F_q^{(n-k-\ell) \times (k+\ell)}$ and
  $\Hm'' \in \F_q^{\ell \times (k+\ell)}$ such that:
  \[
  \Sm\Hmpi = 
  \begin{pmatrix} 
    \un_{n-k-\ell} & \Hm' \\ 
    \mathbf{0} & \Hm''
  \end{pmatrix}.
  \]
 
  The error $\ev$ can be written as $\ev = (\ev',\ev'')$ where
  $\ev' \in \F_q^{n-k-\ell}$ and $\ev'' \in \F_q^{k+\ell}$, and one
  can write $\sv\transpose{\Sm} = (\sv',\sv'')$ with
  $\sv' \in \F_q^{n-k-\ell}$ and $\sv'' \in \F_q^{\ell}$.

  \begin{align*} 
    \Hmpi\transpose{\ev} = \transpose{\sv} 
    & \iff \Sm\Hm_\pi\transpose{\ev} = \Sm\transpose{\sv} \\
    & \iff 
      \begin{pmatrix} 
        \un_{n-k-\ell} & \Hm' \\ 
        \mathbf{0} & \Hm''
      \end{pmatrix} 
      \begin{pmatrix} 
        \transpose{\ev'} \\
        \transpose{\ev''} 
      \end{pmatrix} 
      = 
      \begin{pmatrix} 
        \transpose{\sv'} \\
        \transpose{\sv''} 
      \end{pmatrix} \\ 
    &\iff 
      \left\{
      \begin{array}{ll} 
        \transpose{\ev'} + \Hm'\transpose{\ev''} = \transpose{\sv'} \\ 
        \Hm''\transpose{\ev''} = \transpose{\sv''}
      \end{array} 
      \right. 
      \numberthis \label{EQ1}
  \end{align*} 

  To solve the problem, we will try to find a solution $(\ev',\ev'')$
  to the above system such that $|\ev''| = p$ and $|\ev'| = w-p$.

\item \emph{The Subset Sum step.}  Compute a set
  $\cS \subseteq \F_q^{k + \ell}$ of solutions $\ev''$ of
  $\Hm''\transpose{\ev''} = \transpose{\sv''}$ such that
  $|\ev''| = p$. We will solve this problem by considering it as a
  Subset Sum problem as it is described in Subsection
  \ref{subsec:redSubsetSum}.

\item \emph{The test step.} Take a vector $\ev'' \in \cS$ and let
  $ \transpose{\ev'} = \transpose{\sv'} -
  \Hm'\transpose{\ev''}$.
  Equation \eqref{EQ1} ensures that
  $\Hmpi\transpose{(\ev',\ev'')} = \transpose{\sv}$.  If
  $|\ev'| = w-p$, $\ev = (\ev',\ev'')$ is a solution of $\SD(q,R,W)$
  on inputs $\Hmpi$ and $\sv$, which can be turned into a solution of
  the initial problem by permuting the indices, as detailed in
  Equation \eqref{eq:permutation}. Else, try again for other values of
  $\ev'' \in \cS$. If no element of $\cS$ gives a valid solution, go
  back to step $1$.
\end{enumerate}

At the end of protocol, we have a vector $\ev$ such that
$\Hmpi \transpose{\ev} = \transpose{\sv}$ and $|\ev| = w$. Let
$\ev_{\pi^{-1}}$ be the vector $\ev$ where we permute all the
coordinates according to $\pi^{-1}$. Hence,
\begin{equation}
\Hm \transpose{\ev}_{\pi^{-1}} = \Hmpi \transpose{\ev} = \transpose{\sv} 
\quad \textrm{ and } \quad 
|\ev_{\pi^{-1}}| = |\ev| =  w.
\label{eq:permutation}
\end{equation}
Therefore, $\ev_{\pi^{-1}}$ is a solution to the problem.

\subsection{Analysis of the Algorithm} 

In order to analyse this algorithm, we rely on the following two
propositions.

\begin{notation}
  An important quantity to understand the complexity of this algorithm
  is the probability of success at step 4. On an input $(\Hm,\sv)$
  uniformly drawn at random, suppose that we have a solution to the
  Subset Sum problem, \textit{i.e.} a vector $\ev''$ such that
  $\Hm''\transpose{\ev''} = \transpose{\sv''}$ and $|\ev''| = p$. Let
  $\transpose{\ev'} = \transpose{\sv'} - \Hm'\transpose{\ev''}$.
  We will denote:
  \[\Ppl \eqdef \mathbb{P}\left(|\ev'| = w - p \mid 
    |\ev''| = p\right).\]
\end{notation}

\begin{proposition} \label{propo:wkfac}
  We have, up to a polynomial factor,
  \[
  \Ppl = \frac{\binom{n-k-\ell}{w-p}(q-1)^{w-p}} {\min
  \left( q^{n-k-\ell},\binom{n}{w}(q-1)^{w}q^{-\ell} \right)}.
  \]
\end{proposition}

\begin{proof} 
  The proof of this statement is simple combinatorics. The numerator
  corresponds to the number of vectors $\ev'$ of weight $w-p$. The
  denominator corresponds to the inverse of the probability that
  $\transpose{\ev'} = \transpose{\sv'} - \Hm'\transpose{\ev''}$. For a
  typical random behavior, this is equal to $q^{n-k-\ell}$. But here
  we know that there is at least one solution. Therefore, we know that
  the number of vectors of weight $w-p$ is bounded from above by the
  number of vectors $\ev$ such that
  $\Hm''\transpose{\ev''} = \transpose{\sv''}$. This explains the
  second term of the minimum. \qed
\end{proof}

\begin{proposition} 
  Assume that we have an algorithm that finds a set $\cS$ of solutions
  of the Subset Sum problem in time $T$. The average running time of the
  algorithm is, up to a polynomial factor,
  \[
  T \cdot \max\left(1,\frac{1}{|\cS| \cdot \Ppl}\right).
  \]
\end{proposition}

As we can see, all the parameters are entwined. The success probability
$\Ppl$ depends of $p$ and $\ell$, as well as the time $T$ to find the set 
$\cS$ of solutions.

In this work, we will focus on a family of parameters useful in the
analysis of the {\wave} signature scheme \cite{DST18}. More precisely,
we will study the following regime:
\[ q=3 \quad ; \quad R \in [0.5,0.9]  \quad ; \quad W \in [0.9,0.99].\]

One consequence of working with a very high relative weight $W$ is
that our best algorithms will work with:
\begin{equation} \label{eq:regime}
\ell = \Theta(n) \quad ; \quad p = k+\ell.
\end{equation}

Here, $\ell$ is $\Theta(n)$ for the following reason: if $\ell = o(n)$
then it is readily verified that, asymptotically in $n$, the average
running time of the \PGESS framework will be bounded from below (up to
a polynomial factor) by $1/\mathcal{P}_{p,0}$. This exactly
corresponds to the complexity of the simplest generic algorithm to
solve $\SD$, namely Prange's ISD algorithm \cite{P62}.

\subsection{Reduction to the Subset Sum Problem} 
\label{subsec:redSubsetSum}

In step $3$ of the \PGESS framework, we have a matrix
$\Hm'' \in \F_q^{\ell \times (k+\ell)}$, a vector
$\sv'' \in \F_q^{\ell}$ and we want to compute a set
$\cS \subseteq \F_q^{k + \ell}$ of solutions $\ev''$ of
$\Hm''\transpose{\ev''} = \transpose{\sv''}$ such that $|\ev''| = p$.
At first sight, this looks exactly like a Syndrome Decoding problem
with inputs $\Hm''$ and $\sv''$ so we could just recursively apply the
best $\SD$ algorithm on this subinstance. But the main difference is
that, in this case, we want to find many solutions to the problem and
not just one. One possibility to solve this problem is to reduce it to
the Subset Sum problem on vectors in $\F_q^{\ell}$.

\begin{restatable}{problem}
  {subsetsum}[Subset Sum problem - $\SS(q,n,m,L,p)$]
  \label{prob:SS}
  
  \begin{tabular}{ll}
    Instance: \quad & 
                      $n$ vectors $\xv_{i} \in \F_q^m$ 
                      for $1 \leq i \leq n$, 
                      a target vector $\sv \in \F_q^m$. \\
    Output: & $L$ solutions $\bv^{(j)} = (b^{(j)}_1,\dots,b^{(j)}_n) 
              \in \{0,1\}^{n}$ for $1 \leq j \leq L$, \\
                    & such that for all $j$, 
                      $\sum_{i=1}^n b^{(j)}_i \xv_{i} = 
                      \sv$ and $\wt{\bv^{(j)}} = p$.
  \end{tabular}
\end{restatable}

We can consider the same problem with elements $b$ in $\F_q$ instead
of $\zo$.

\begin{restatable}{problem}
  {subsetsumNZC}
  [Subset Sum with non-zero characteristic -
 $\SSNZC(q,n,m,L,p)$]
  \label{prob:SSNZC}

  \begin{tabular}{ll} 
    Instance: \quad & $n$ vectors $\xv_{i} \in \F_q^m$ 
                      for $1 \leq i \leq n$, 
                      a target vector $\sv \in \F_q^m$. \\
    Output: & $L$ solutions $\bv^{(j)} = (b^{(j)}_1,\dots,b^{(j)}_n) 
              \in \F_q^{n}$ for $1 \leq j \leq L$, \\
                    & such that for all $j$, 
                      $\sum_{i=1}^n b^{(j)}_i \xv_{i} = \sv$ 
                      and $\wt{\bv^{(j)}} = p$.
  \end{tabular}
\end{restatable}

\begin{notation}
  We will denote $\SS(q,n,m,L,\emptyset)$ (\textit{resp.}
  $\SSNZC$) the $\SS$ problem (\textit{resp.}
  $\SSNZC$ problem) without any constraint on the weight.
\end{notation}

Again, we will be interested in the average case, where all the inputs
are taken uniformly at random. Notice that the problem that needs to
be solved at step $3$ of the PGE+SS framework reduces exactly to
$\SSNZC(q,k+\ell,\ell,|\cS|,p)$.

There is an extensive literature \cite{HJ10,BCJ11} about the Subset
Sum problem for specific parameter ranges, typically when
$L = 1, q = 2, n = m$ and $p = \frac{m}{2}$. This is the hardest case
where there is on average a single solution. There are several regimes
of parameters, each of which lead to different algorithms. For
instance, when $m = O(n^{\varepsilon})$ for $\varepsilon < 1$, there
are many solutions on average and we are in the high density setting
for which we have sub-exponential algorithm \cite{L05}.  Table
\ref{table:SS} summarizes the complexity of algorithms to solve the
Subset Sum problem for some different regimes of parameters when only
one solution is required ($L=1$) and for $q = 2$.

\begin{table}[h!] 
  \centering
  \begin{tabular}{ c c c }
    \toprule
    \quad Value of $m$ \quad & \quad Complexity \quad & \quad 
                                                        Reference \quad \\
    \midrule
    $O(\log(n))$ & $\text{poly}(n)$ & \cite{GM91,CFG89} \\
    $O(\log(n)^{2})$ & $\text{poly}(n)$ & \cite{FP05}  \\
    $O(n^{\varepsilon})$ 
    for $\varepsilon < 1$ & 
                            $ 2^{O\left(\frac{n^{\varepsilon}}{\log(n)}\right)}$ 
                              & \cite{L05} \\
    $n$ & $2^{O(n)}$ & \cite{HJ10,BCJ11} \\
    \bottomrule 
  \end{tabular}
  
  \vspace*{+0.3cm}
  \caption{Complexity of best known algorithms to solve 
$\SS(2,n,m,1,\emptyset)$. \label{table:SS}}
\end{table}

In our case, $m$ will be a small, but constant, fraction of $n$, which
leads to multiple solutions but exponentially complex algorithms to
find them. We will be in a moderate density situation. Furthermore,
the case $L = 1$ and $L \gg 1$ require quite different
algorithms. When $q = 2$, authors of \cite{BJMM12} show how to
optimize this whole approach to solve the original Syndrome Decoding
problem using better algorithms for the Subset Sum problem.

\subsection{Application to the \PGESS Framework with High Weight}

There are quite a lot of interesting regimes that could be studied
with this approach and have not been studied yet. Indeed, very few
papers tackle the case $q \geq 3$ and they only cover a small fraction
of the possible parameters.  In this work we focus on the problem
$\SSNZC(3,k+\ell,\ell,|\cS|,k + \ell)$ given by the \PGESS framework
for high weights in $\F_3$. The choice of $p = k + \ell$ for large weights is
explained in Equation \eqref{eq:regime}. This is quite convenient
because this problem is actually equivalent to solving
$\SS(3,k+\ell,\ell,|\cS|,\emptyset)$ as shown by the following lemma.

\begin{lemma}
  \label{lem:reduction}
  If we have an algorithm that solves
  $\SS(3,k+\ell,\ell,|\cS|,\emptyset)$ then we have an algorithm that solves
  $\SSNZC(3,k+\ell,\ell,|\cS|,k+\ell)$ with the same complexity.
\end{lemma}

\begin{proof}
  Let $\mathcal{A}$ be an algorithm that solves
  $\SS(3,k+\ell,\ell,|\cS|,\emptyset)$ and consider an instance
  $(\xv_{1},\dots,\xv_{k+\ell})$, $\sv$ of
  $\SSNZC(3,k+\ell,\ell,|\cS|,k+\ell)$.  We want to find
  $b_1,\dots,b_{k+\ell} \in \{1,2\}$ (see $\F_3 = \{0,1,2\}$) such that
  $\sum_{i=1}^{k+\ell} b_i \xv_{i} = s$. Let
  $\sv' = 2\sv + \sum_i \xv_{i}$ and let us run $\mathcal{A}$ on input
  $(\xv_{1},\dots,\xv_{k+\ell}), \sv'$. We obtain
  $b'_1,\dots,b'_{k+\ell} \in \zo$ such that
  $\sum_{i=1}^{k+\ell} b'_i \xv_{i} = \sv'$. Take
  $b_i = \frac{b'_i-1}{2}$ for $1 \leq i \leq k+\ell$, where the
  division is done in $\F_3$ and return $(b_1,\dots,b_{k+\ell})$.

  Indeed, this gives a valid solution to the problem: the elements
  $b_i$ belong to $\{1,2\}$ and we have:
  \[
  \sum_{i=1}^{k+\ell} b_i \xv_{i} = \sum_{i=1}^{k+\ell} \frac{
    b'_i-1}{2}\xv_{i} = \frac{\sv'}{2} - \frac{\sum_{i=1}^{k+\ell}
    \xv_{i}}{2} = \sv.  
  \]   

\vspace{-10mm}
\qed

\vspace{5mm}
\end{proof}

Hence, in the context of the \PGESS framework for solving $\SD$ with
high weights, it is enough to solve
$\SS(3,k+\ell,\ell,|\cS|,\emptyset)$. However, as explained at the end
of Subsection \ref{subsec:PGESS}, we will have to choose
$\ell = \Theta(n) = \Theta(k)$ (because $k = \lceil R n \rceil$).
Therefore, we are in a regime where solving the Subset Sum problem
requires exponential complexity, as explained in the previous
subsection. However, as we will see in the next session, we will be
able to choose $\ell$ as a small fraction of $k$. In this case,
generic algorithms as Wagner's \cite{W02} perform exponentially better
compared to Prange's algorithm \cite{P62} (case $\ell = 0$) or Subset
Sum algorithms \cite{BCJ11} (case $\ell = n-k$).

%% file: 3-GBA.tex
We show in this section how to solve $\SS(3,k+\ell,\ell,L,\emptyset)$,
first with Wagner's algorithm \cite{W02}. Parameters $k$ and $\ell$
will be free. We will focus on the values $L$ for which we can find
$L$ solutions to $\SS\left(3,k+\ell,\ell,L,\emptyset\right)$ in time
$O(L)$.  In such a case, we say that we can find solutions in
\emph{amortized time $O(1)$}.

\subsection{A Brief Description of Wagner's Algorithm}

Recall that we are here in the context of the Subset Sum step of the
PGE+SS framework described in Subsection \ref{subsec:PGESS}. Given
$k+\ell$ vectors $\xv_{1},\cdots,\xv_{k+\ell} \in \F_3^{\ell}$
(columns of the matrix $\Hm''$) and a target vector
$\sv \in \F_3^{\ell}$, our goal is to find $L$ solutions of the form
$\bv^{(j)} = (b^{(j)}_1,\cdots,b^{(j)}_{k+\ell}) \in \{0,1\}^{k+\ell}$
such that for all $1 \leq j \leq L$,
\begin{equation} 
\label{eq:wagEq} 
\sum_{i=1}^{k+\ell} b^{(j)}_i \xv_{i} = \sv.
\end{equation} 

Here, we are interested in the average case, which means that all the
vectors $\xv_{i}$ are independent and follow a uniform law over
$\F_3^{\ell}$. In order to apply Wagner's algorithm \cite{W02}, let
$a \in \mathbb{N}^{*}$ be some integer parameter. For
$i \in \llbracket 1, 2^{a} \rrbracket$, denote by $\Ic_i$ the sets
$\Ic_i \eqdef\ \llbracket 1 +
\frac{(i-1)({k+l})}{2^a},\frac{i({k+l})}{2^a} \rrbracket$.
The sets $\Ic_i$ form a partition of $\llbracket1,{k+\ell}\rrbracket$.

The first step of Wagner's algorithm is to compute $2^{a}$ lists
$(\Lc_i)_{1 \leq i \leq 2^a}$ of size $L$ such that:
\begin{equation}
\label{eq:listSS}  
\forall i \in \llbracket 1, 2^{a} \rrbracket, \mbox{ } 
\Lc_{i} \subseteq \left\{ \sum_{j \in \Ic_i} b_j\xv_{j} \mbox{ } : \mbox{ } 
\forall j \in \Ic_i, \; b_j \in \zo \right\} \textrm{ and } | \Lc_i | = L.
\end{equation} 

Each list $\Lc_i$ consists of $L$ random elements of the form
$\sum_{j \in \Ic_i} b_j\xv_{j}$ where the randomness is on
$b_j \in \{0,1\}$. By construction, we make sure that given
$\yv \in \Lc_i$ we have access to the coefficients
$(b_j)_{j \in \Ic_i}$ such that $\yv = \sum_{j \in \Ic_i}b_j \xv_{j}$.
In other words, we have divided the vectors
$\xv_1,\ldots,\xv_{k+\ell}$ in $2^{a}$ stacks of $(k+\ell)/2^{a}$
vectors and for each stack we have computed a list of $L$ random
linear combinations of the vectors in the stack. The running time to build
theses lists is $O(L)$. 
Once we have computed these lists we can use the main idea of Wagner
to solve \eqref{eq:wagEq}. In our case we would like to find solutions
in amortized time $O(1)$. For this, Wagner's algorithm requires the lists
$\Lc_i$ to be all of the same size:
\[
\forall i \in \llbracket 1, 2^{a} \rrbracket, \ 
|\Lc_i| = L = 3^{\ell/a}.
\]
This gives a first constraint on the parameters $k,\ell$ and $a$, namely:
\[
3^{\ell/a} \leq 2^{(k+\ell)/2^{a}} \quad (\mbox{number of vectors }
\bv^{(j)} \mbox{ in each stack}).
\]
which  puts a constraint on $a$ since $k,\ell$ are fixed. 
With these lists at hand, Wagner's idea is to merge the lists in the following way.
For every $p \in \{1,3,\cdots,2^{a}-3\}$, create a list $\Lc_{p,p+1}$ from 
 $\Lc_{p}$ and $\Lc_{p+1}$ such that:
\[
\Lc_{p,p+1} \eqdef \left\{ \yv_p + \yv_{p+1} \; : \; \yv_i \in
  \Lc_i \textrm{ and the last $\ell/a$ bits of $\yv_p+\yv_{p+1}$ are } 0\mbox{s.} \right\}.
\]
A list $\Lc_{2^a-1,2^a}$ is created from $\Lc_{2^{a}-1}$ and
$\Lc_{2^{a}}$ in the same way except that the last $\ell/a$ bits have
to be equal to those of $\sv$. As the elements of the lists $\Lc_p$
are drawn uniformly at random in $\F_3^{\ell}$, it is easily verified
that by merging them on $\ell/a$ bits, the new lists $\Lc_{p,p+1}$ are
typically of size ${|\Lc_i|^{2}}/3^{\ell/a} = {(3^{\ell/a})^{2}}/{3^{\ell/a}} = 3^{\ell/a}$.
Therefore, the cost in time and in space of such a merging (by using
classical techniques such as hash tables or sorted lists) will be
$O(3^{\ell/a})$ on average. This way, we obtain $2^{a-1}$ lists of
size $L$. It is readily seen that we can repeat this process $a-1$
times, with each time a cost of $O(3^{\ell/a})$ for merging on
$\ell/a$ new bits. After $a$ steps, we obtain a list of solutions to the
Equation \eqref{eq:wagEq} containing $L = 3^{\ell/a}$ elements on
average.

	\begin{figure}[h!]
		\centering
		\begin{tikzpicture}
		\draw [<-] (0.5,0.5) -- (2,0.5);
		\node at (3.2,0.5) {Set of solutions};
		\draw (-0.35,0) rectangle (0.35,1);
		\draw [->] (-1.25,-1.5) -- (-0.05,-0.3);
		\draw [->] (1.25,-1.5) -- (0.05,-0.3);
		\draw [<-] (1.25,-0.9) -- (2.25,-0.9);
		\node at (4.8,-0.9) {Merging on $\ell/2$ bits according to $\sv$};
		\draw (0.9,-2.8) rectangle (1.6,-1.8) ;
		\node at (1.25,-2.55) {$\sv_{\ell/2}$};
		\draw [<->] (1.8,-2.8) -- (1.8,-2.3);
		\node at (2.15,-2.55) {$\ell/2$};
		\draw (0.9,-2.8) rectangle (1.6,-2.3);
		\draw (-1.6,-2.8) rectangle (-0.9,-1.8);
		\node at (-1.25,-2.55) {$\mathbf{0}$};
		\draw [<->] (-1.8,-2.8) -- (-1.8,-2.3);
		\node at (-2.15,-2.55) {$\ell/2$};
		\draw (-1.6,-2.8) rectangle (-0.9,-2.3);
		\draw (-1.4,-3.6) rectangle (-2.1,-4.6);
		\node at (-1.75,-5) {$\Lc_1$};
		\draw [->] (-2,-3.5) -- (-1.3,-2.9);
		\draw (-0.4,-3.6) rectangle (-1.1,-4.6);
		\node at (-0.75,-5) {$\Lc_2$};
		\draw [->] (-0.5,-3.5) -- (-1.2,-2.9);
		\draw (1.1,-3.6) rectangle (0.4,-4.6);
		\node at (0.75,-5) {$\Lc_3$};
		\draw [->] (0.5,-3.5) -- (1.2,-2.9);
		\draw (2.1,-3.6) rectangle (1.4,-4.6);
		\node at (1.75,-5) {$\Lc_4$};
		\draw [->] (2,-3.5) -- (1.3,-2.9);
		\node at (5.5,-3.3) {Merging on $\ell/2$ bits according to $\sv$};
		\draw [<-] (2,-3.3)--(3,-3.3);
		\draw [->] (-3,-3.3)--(-2,-3.3);
		\node at (-4.5,-3.3){Merging on $\ell/2$ bits};
		\end{tikzpicture}
		\caption{Wagner's algorithm with $a = 2$.} 
		\label{fig:GBA4} 
	\end{figure}
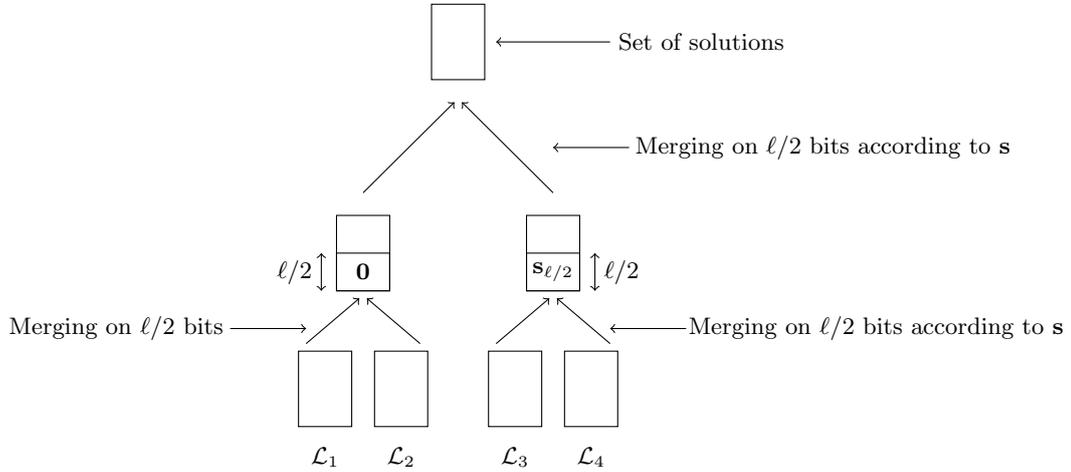 

Let us summarize the previous discussion with the following theorem.  

\begin{theorem} 
  \label{Theorem:Wagner} 
  Fix $k,\ell \in \mathbb{N}^*$ and let $a$ be any non zero integer
  such that
  \begin{equation*}
    3^{\ell/a} \leq 2^{(k+\ell)/2^{a}}.
  \end{equation*}

  The associated $\SS(3,k+\ell,\ell,3^{\ell/a},\emptyset)$ problem can be
  solved in average time and space $O(3^{\ell/a})$.
\end{theorem}

This theorem indicates for which value $L$ it is possible to find $L$
solutions in time $O(L)$ using Wagner's approach.

\subsection{Smoothing of Wagner's Algorithm}
\label{subsec:smoothing}

Wagner's algorithm as stated above shows how to find $L$ solutions in
amortized time $O(1)$ for $L = 3^{\ell/a}$. If we want more than $L$
solutions, we can repeat this algorithm and find all those solutions
also in amortized time $O(1)$. So the smaller $L$ is, the better the
algorithm performs. So the idea is to take the largest integer $a$
such that $3^{\ell/a} < 2^{(k+\ell)/2^{a}}$ and take $L = 3^{\ell/a}$,
as explained in Theorem \ref{Theorem:Wagner}.  But this induces a
discontinuity in the optimal value of $L$ and on the complexity: when
the input parameters change continuously, the optimal value of $a$
(which has to be an integer) evolves discontinuously, therefore the
slope of the complexity curve is discontinuous, as we can see on
Figures \ref{fig:complexity_R5} and \ref{fig:complexity_R676}. We show
here a refinement of Theorem \ref{Theorem:Wagner} that reduces the
discontinuity.

\begin{proposition} 
  Let $a$ be the largest integer such that
  $3^{\ell/(a-1)} < 2^{(k+\ell)/2^{a-1}}$. If $a \ge 3$, the above
  algorithm can find $2^{\lambda}$ solutions in time $O(2^{\lambda})$ with
  \[\lambda = {\frac{\ell \log(3)}{a-2} -
      \frac{k+\ell}{(a-2)2^{a-1}}}.\]
  \label{Prop:RealSmoothing}
\end{proposition} 

We see that we retrieve the result of Theorem \ref{Theorem:Wagner}
when $3^{\ell/a} = 2^{(k+\ell)/2^{a}}$. We have not found any
statement of this form in the literature, which is surprising because
Wagner's algorithm has a variety of applications. We now prove the
proposition.

\begin{proof}
  Parameters $k$ and $\ell$ are fixed. Let $a$ be the largest integer
  such that $3^{\ell/(a-1)} < 2^{(k+\ell)/2^{a-1}}$ and we suppose
  that $a \ge 3$.  We will consider Wagner's algorithm on $a$ levels
  but the merging at the bottom of the tree will be performed with a
  lighter constraint: we want the sums to agree on less than $\ell/a$
  bits. Indeed, we consider the following list sizes. At the bottom of
  the trees, we take lists of size $2^{\frac{k+\ell}{2^a}}$ (the
  maximal possible size); at all other levels, we want lists of size
  $ 2^\lambda$. We run Wagner's algorithm by firstly merging on $m$
  bits. In order to obtain lists of size $2^{\lambda}$ at the second
  step, we have to choose $m$ such that
  \begin{equation}
    \label{eq:EQ6}
    \frac{\left(2^{(k+\ell)/2^a}\right)^2}{3^m} = 2^{\lambda}
    \qquad \textrm{ \textit{i.e.} \qquad}
    \frac{2(k+\ell)}{2^a} - m \log_2(3) = \lambda.
  \end{equation}
  The other $(a-1)$ merging steps are designed such that merging two
  lists of size $2^{\lambda}$ gives a new list of size $2^{\lambda}$,
  which means that we merge on $\lambda/\log_2(3)$ bits. However, in
  the final list we want to obtain solutions to the problem, which
  means that in total we have to put a constraint on all bits.
  Therefore, $\lambda$ and $m$ have to verify:
  \begin{equation} \label{eq:mlambd}
  m + (a-1)\frac{\lambda}{\log_2(3)} = \ell. 
  \end{equation} 
  By combining Equations \eqref{eq:EQ6} and \eqref{eq:mlambd} we get:
  \[
  \lambda = \frac{\ell\log_2(3)}{a-2} - \frac{k+\ell}{(a-2)2^{a-1}}\cdot
  \]
  It is easy to check that under the conditions
  $3^{\ell/(a-1)} < 2^{(k+\ell)/2^{a-1}}$ and $a \geq 3$, $\lambda$
  and $m$ are positive which concludes the proof. 
  \qed
\end{proof}

%\begin{figure}[htb!]
%	\caption{Exponents of the generalized birthday algorithm against
%		the best generic algorithms to solve the
%		syndrome decoding problem with large weights
%		when $k/n = 0.755$}
%	\centering
%	\includegraphics[width=10cm]{img/exp1}
%\end{figure}

%% file: 4-representations.tex
\subsection{Basic Idea}
\label{subsec:representations}

In the list tree of Wagner's algorithm (see Figure \ref{fig:GBA4}), we
split each list in two, according to what is called the \emph{left-right}
procedure. This means that if we start from a set
$ S = \{\sum_{j \in \llbracket A,B \rrbracket } b_j\xv_j : |b_j| =
p\}$,
we decompose each element of $\yv \in S$ as $\yv = \yv_1 + \yv_2$
where $\yv_1 \in S_1$ and $\yv_2 \in S_2$, where
\begin{align*} 
S_1 & \eqdef \left\{\sum_{j \in \llbracket A,\lfloor\frac{B+A}{2}\rfloor \rrbracket} 
      b_j\xv_j : b_j \in \zo, \ |\bv| = p/2
      \right\} \\ 
S_2 & \eqdef \left\{\sum_{j \in \llbracket \lfloor\frac{B+A}{2}\rfloor+1,B \rrbracket} 
      b_j\xv_j : b_j \in \zo, \ |\bv| = p/2
      \right\}.
\end{align*} 
Such a decomposition does not always exist, but it exists with
probability at least $\frac{1}{p}$.  Indeed, the probability that a
vector of weight $p$ can be split this way is
\[
\frac{\left.\binom{n/2}{p/2}\right.^{2}}{\binom{n}{p}} \geq \frac{1}{p}.
\]

Wagner's algorithm uses this principle. When looking for vectors
$\mathbf{b}$ containing the same number of $0$'s and $1$'s, it looks for
$\mathbf{b}$ in the form $\mathbf{b} = \mathbf{b}_1 + \mathbf{b}_2$,
where the second half of $\mathbf{b}_1$ and the first half of
$\mathbf{b}_2$ are only zeros. The first half of $\mathbf{b}_1$ and
the second half of $\mathbf{b}_2$ are expected to have the same number
of $0$s and $1$s.

The idea of representations is to follow Wagner's approach of list
merging while allowing more possibilities to write $\bv$ as the sum
of two vectors $\bv = \bv_1 + \bv_2$. We remove the constraint that
$\bv_1$ has zeros on its right half and $\bv_2$ has zeros on its left
half.  We replace it by a less restrictive constraint: we fix the
number of $0$s, $1$s and $2$s (see $\mathbb{F}_3 = \{0,1,2\}$) in
$\bv_1$ and $\bv_2$.

\begin{figure}[h!]
\parbox{.45\linewidth}{
\centering
\begin{tabular}{cccccccc} \hline 
\multicolumn{1}{|c|}{\;1\;} &
\multicolumn{1}{c|}{\;0\;} & 
\multicolumn{1}{c|}{\;0\;} &
\multicolumn{1}{c|}{\;1\;} &
\multicolumn{1}{c|}{\cellcolor{gray!20}0} &
\multicolumn{1}{c|}{\cellcolor{gray!20}0} &
\multicolumn{1}{c|}{\cellcolor{gray!20}0} &
\multicolumn{1}{c|}{\cellcolor{gray!20}0} \\ \hline & & &
\multicolumn{2}{c}{+} & & & \\ \hline
\multicolumn{1}{|c|}{\cellcolor{gray!20}0} &
\multicolumn{1}{c|}{\cellcolor{gray!20}0} &
\multicolumn{1}{c|}{\cellcolor{gray!20}0} &
\multicolumn{1}{c|}{\cellcolor{gray!20}0} &
\multicolumn{1}{c|}{\;0\;} & 
\multicolumn{1}{c|}{\;1\;} &
\multicolumn{1}{c|}{\;0\;} & 
\multicolumn{1}{c|}{\;1\;} \\ \hline & & &
\multicolumn{2}{c}{=} & & & \\ \hline 
\multicolumn{1}{|c|}{\cellcolor{red!20}\;1\;} &
\multicolumn{1}{c|}{\cellcolor{red!20}\;0\;} & 
\multicolumn{1}{c|}{\cellcolor{red!20}\;0\;} &
\multicolumn{1}{c|}{\cellcolor{red!20}\;1\;} & 
\multicolumn{1}{c|}{\cellcolor{red!20}\;0\;} &
\multicolumn{1}{c|}{\cellcolor{red!20}\;1\;} & 
\multicolumn{1}{c|}{\cellcolor{red!20}\;0\;} &
\multicolumn{1}{c|}{\cellcolor{red!20}\;1\;} \\ \hline
\end{tabular}
}
\hfill
\parbox{.45\linewidth}{
\centering
\begin{tabular}{cccccccc} \hline 
\multicolumn{1}{|c|}{\;1\;} &
\multicolumn{1}{c|}{\;0\;} & 
\multicolumn{1}{c|}{\;2\;} &
\multicolumn{1}{c|}{\;0\;} &
\multicolumn{1}{c|}{\;0\;} &
\multicolumn{1}{c|}{\;1\;} &
\multicolumn{1}{c|}{\;1\;} &
\multicolumn{1}{c|}{\;0\;} \\ \hline & & &
\multicolumn{2}{c}{+} & & & \\ \hline
\multicolumn{1}{|c|}{\;0\;} &
\multicolumn{1}{c|}{\;0\;} &
\multicolumn{1}{c|}{\;1\;} &
\multicolumn{1}{c|}{\;1\;} &
\multicolumn{1}{c|}{\;0\;} & 
\multicolumn{1}{c|}{\;0\;} &
\multicolumn{1}{c|}{\;2\;} & 
\multicolumn{1}{c|}{\;1\;} \\ \hline & & &
\multicolumn{2}{c}{=} & & & \\ \hline 
\multicolumn{1}{|c|}{\cellcolor{red!20}\;1\;} &
\multicolumn{1}{c|}{\cellcolor{red!20}\;0\;} & 
\multicolumn{1}{c|}{\cellcolor{red!20}\;0\;} &
\multicolumn{1}{c|}{\cellcolor{red!20}\;1\;} & 
\multicolumn{1}{c|}{\cellcolor{red!20}\;0\;} &
\multicolumn{1}{c|}{\cellcolor{red!20}\;1\;} & 
\multicolumn{1}{c|}{\cellcolor{red!20}\;0\;} &
\multicolumn{1}{c|}{\cellcolor{red!20}\;1\;} \\ \hline
\end{tabular}

\vspace{8mm}

\begin{tabular}{cccccccc} \hline 
\multicolumn{1}{|c|}{\;1\;} &
\multicolumn{1}{c|}{\;0\;} & 
\multicolumn{1}{c|}{\;1\;} &
\multicolumn{1}{c|}{\;1\;} &
\multicolumn{1}{c|}{\;0\;} &
\multicolumn{1}{c|}{\;0\;} &
\multicolumn{1}{c|}{\;2\;} &
\multicolumn{1}{c|}{\;0\;} \\ \hline & & &
\multicolumn{2}{c}{+} & & & \\ \hline
\multicolumn{1}{|c|}{\;0\;} &
\multicolumn{1}{c|}{\;0\;} &
\multicolumn{1}{c|}{\;2\;} &
\multicolumn{1}{c|}{\;0\;} &
\multicolumn{1}{c|}{\;0\;} & 
\multicolumn{1}{c|}{\;1\;} &
\multicolumn{1}{c|}{\;1\;} & 
\multicolumn{1}{c|}{\;1\;} \\ \hline & & &
\multicolumn{2}{c}{=} & & & \\ \hline 
\multicolumn{1}{|c|}{\cellcolor{red!20}\;1\;} &
\multicolumn{1}{c|}{\cellcolor{red!20}\;0\;} & 
\multicolumn{1}{c|}{\cellcolor{red!20}\;0\;} &
\multicolumn{1}{c|}{\cellcolor{red!20}\;1\;} & 
\multicolumn{1}{c|}{\cellcolor{red!20}\;0\;} &
\multicolumn{1}{c|}{\cellcolor{red!20}\;1\;} & 
\multicolumn{1}{c|}{\cellcolor{red!20}\;0\;} &
\multicolumn{1}{c|}{\cellcolor{red!20}\;1\;} \\ \hline
\end{tabular}
}

\vspace{2mm}
\parbox{.45\linewidth}{\centering (1)}
\hfill
\parbox{.45\linewidth}{\centering (2)}
\caption{Same vector (1) using left-right split and (2) using representations.} 
\label{fig:representations}
\end{figure}

More precisely, we consider the set
\begin{align} \label{Eq:Reps}
S' = \left\{\sum_{j \in \llbracket A,B \rrbracket } b_j\xv_j : b_j \in
\mathbb{F}_3, \ |\{b_j=1\}| = p_1 \mbox{ and } |\{b_j=2\}| = p_2 \right\}
\end{align}
for some weights $p_1$ and $p_2$ and we want to decompose each $\yv$
into $\yv_1 + \yv_2$ such that $\yv_1,\yv_2 \in S'$. On the example of
Figure \ref{fig:representations}, we have $p = 4$, $p_1=3$ and
$p_2=1$. 

At first sight, this approach may seem unusual. Indeed, except for
very specific values of $p_1$ and $p_2$, the sum $\yv_1 + \yv_2$ will
rarely match the desired weight $p$ to be in $S$. Such a sum $\yv_1 + \yv_2$, which matches the targeted bits for merging but not the weight constraint, will be called \textit{badly-formed}. Those \textit{badly-formed} sums cannot be used for the remaining of the algorithm and must be discarded. However, the
positive aspect is that each element $\yv \in S$ accepts many
decompositions (the so-called \textit{representations})
$\yv_1 + \yv_2$ where $\yv_1,\yv_2 \in S'$. The results from
\cite{HJ10,BCJ11,BJMM12} show that this large number of ways to
represent each element can compensate the fact that most
decompositions do not belong to $S$. One can slightly lower the number
of agreement bits when merging the lists, in order to obtain on
average the desired number of elements in the merged list.

Notice that in this definition of $S'$, the elements $b_j$ belong to
the set $\F_3$ and not $\zo$, even though we want to obtain a binary
solution. The ternary structure also increases the number of
representations as shown in Figure \ref{fig:representations}. It is
actually natural to consider representations of binary strings using
three elements $\{0,1,2\}$, as in \cite{BCJ11}.

\subsection{Partial Representations}

If we relieve too many constraints and allow too many representations
of a solution, it may happen that we end up with multiple copies of
the same solution. In order to avoid this situation, we use
\emph{partial representations}, which is an intermediate approach
between \emph{left-right} splitting and using \emph{representations},
as illustrated in Figure \ref{fig:partial-rep}.

\begin{figure}[!h] \centering
	\includegraphics[width = 7cm]{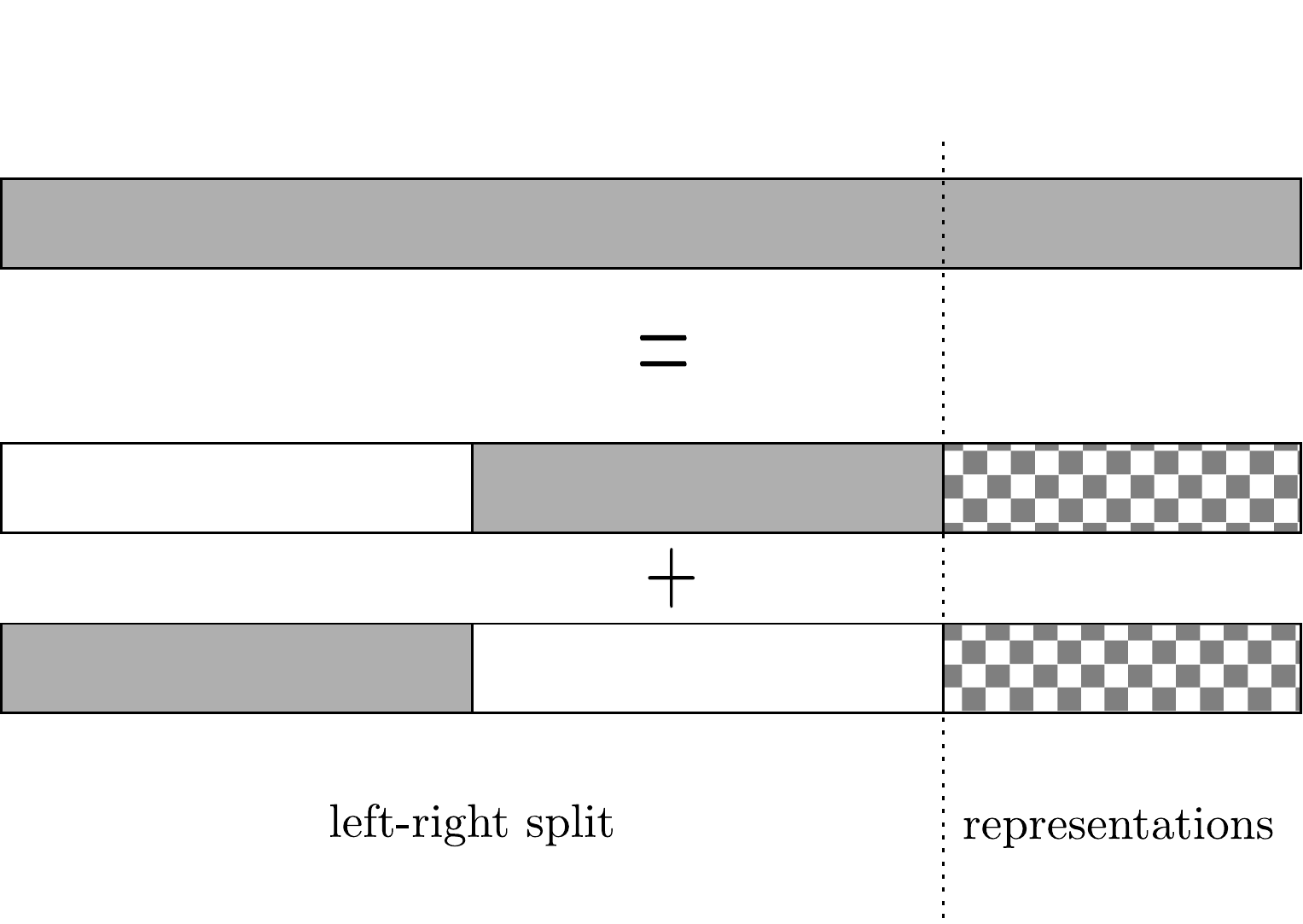}
	\caption{Decomposing a vector using partial representations.}
	\label{fig:partial-rep}
\end{figure}

\subsection{Presentation of our Algorithm} 

Plugging representations in Wagner's algorithm can be done in a
variety of ways. The way we achieved our best algorithm was mostly
done by trial and error. We present here the main features of our
algorithm.

\begin{itemize}
\item In the regime we consider, the number of floors $a$ varies from
  $5$ to $7$. Notice that this is quite larger than in other similar
  algorithms and is mostly due to the fact that we have many solutions
  to our Subset Sum problem.
\item Because we want to find many
  solutions, representations become less efficient. Indeed, the fact
  that we obtain many \textit{badly-formed} elements makes it harder to find
  solutions in amortized time $O(1)$ (or even just in small time).
\item However, we show that representations can still be useful. For
  most parameter range, the optimal algorithm consists of a left-right
  split at the bottom level of the tree, then $2$ layers of partial
  representation and from there to the top level, left-right splits
  again.
\end{itemize}

% \matthieu{Ce paragraphe n'a pas de sens.}
% Each list merging has an associated number of bits on which the
% merging occurs. On the bottom floors, we filter as much as the number
% of representations whereas on top, we filter as much as the list
% size. This means that the bottom constraints filter the
% representations and the top constraints filter some amount of
% solutions. We found that this procedure gave the best results.

Figure \ref{fig:big_tree} illustrates an example for $a = 7$.
When we increase the number of floors, we just add some left-right
splits.

\begin{figure}[!h] \centering
	\includegraphics[width = 14cm]{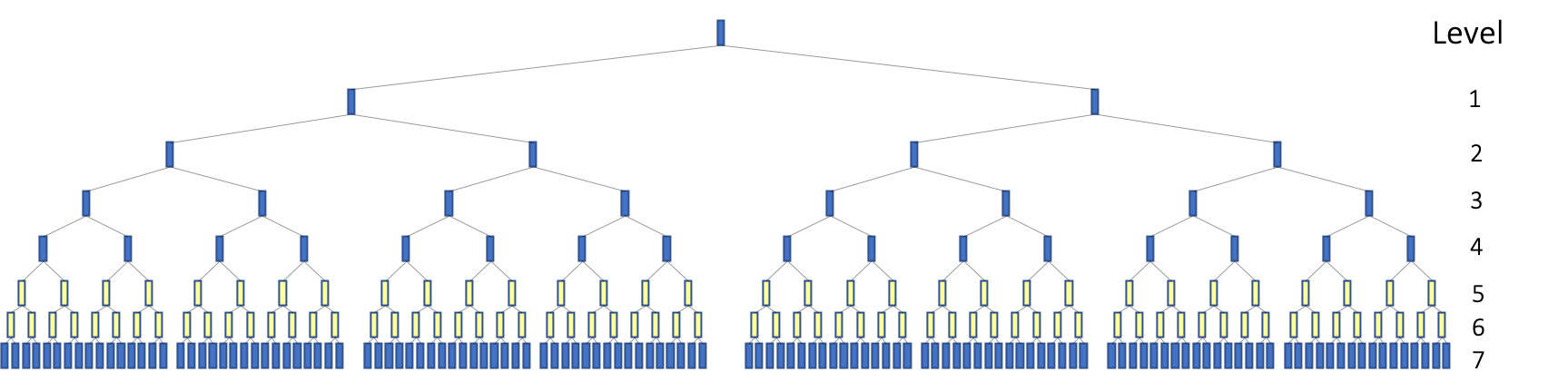}
	\caption{Wagner tree for $a = 7$. Yellow list correspond to
          representations and blue list to left-right splits.}
	\label{fig:big_tree}
\end{figure}

In the next section, we present the different parameters for a
particular input to show how our algorithm behaves.

\subsection{Application to the Syndrome Decoding Problem}
\label{subsec:partial-rep-complexity}

We embedded in the PGE+SS framework the three above-described
algorithms, namely the classical Wagner algorithm, the smoothed one in
\S\ref{sec:GBA} and the one using representation technique in
\S\ref{sec:Reps}. By using Proposition \ref{propo:wkfac}, we derived
the exponents given in Figures \ref{fig:complexity_R5} and
\ref{fig:complexity_R676}.

We present here the details of our algorithm for the SD$(3,R,W)$ with
$R = 0.676$ and $W=0.948366$. These are the parameters which are used
for the analysis of {\wave}. For this set of parameters, we claim that
the complexity of our algorithm is $2^{0.0176n}$.  In the PGE+SS
framework (see Section \ref{subsec:PGESS}), we needed to choose to
parameters $p$ and $\ell$. We take $\ell = 0.060835n$ and
$p = k + \ell$.

The best algorithm we found uses $a=7$, which means that the
associated Wagner tree has $7$ levels, and therefore $128$ leaves
(Figure \ref{fig:big_tree}). From level $0$ to level $6$, the lists
have size $L = 2^{0.0176}$ (\textit{i.e.} equal to the overall
complexity of the Subset Sum problem). As we have more than the
required number of solutions for $6$ levels, but not enough for $7$
levels, we use the smoothing method described in section
\ref{subsec:smoothing}, which gives a size of the leaves equal to
$2^{0.01039}$.

We present below in more detail how we construct the different lists of the Wagner tree.
\begin{itemize}
\item Levels 1 to 4 consist of left-right splits. For instance, at level
$4$, we have $16$ lists
\begin{equation*} \forall i \in \llbracket 1, 16 \rrbracket,
\mbox{ } \Lc_{i} \subseteq \left\{ \sum_{j \in \Ic_i} b_j\xv_{j}
\mbox{ } : \mbox{ } \forall j \in \Ic_i, \; b_j \in \zo \right\}
\textrm{ and } | \Lc_i | = L.
\end{equation*} with $\Ic_i \eqdef\ \llbracket 1 +
\frac{(i-1)({k+\ell})}{16},\frac{i(k+\ell)}{16} \rrbracket$.
\item In levels $5$ and $6$, we use partial representations.   Going from level $4$ to level $5$, on a proportion $\lambda_1 = 0.7252$ of the
vector, we use representations for level $5$ and left-right split for
level $6$. On the remaining fraction of the vector, we use
representations on both levels. More precisely, for each interval $\Ic_i$, we split it in $2$ according to Figure \ref{fig:partial_detail}. For each part, we use Equation \ref{Eq:Reps} with the following densities:
\begin{itemize}
	\item for the part with only one level of representations, $\rho_1$
	consists on $74.8\%$ of $0$s, $25.1\%$ of $1$s and $0.1\%$ of $2$s;
	\item for the part with two levels of representations, we have
	$\rho_2$, composed of $74.2\%$ of $0$s, $25.4\%$ of $1$s and $0.4\%$
	of $2$s for level $5$, and $\rho_3$ composed of $86.9\%$ of $0$s,
	$13.1\%$ of $1$s and $0.0\%$ of $2$s for level $6$.
\end{itemize} 

\item In order to construct level $7$, we start from each list of level $6$ and perform again a left-right split. 
\end{itemize}

The choice of the densities and all the calculi related to the
representations can be quite complicated. We perform a full analysis
in Appendix \ref{sec:ternaryReps}, see in particular Proposition
\ref{Prop:Reps!}.

\begin{figure}[h!] \centering
	\includegraphics[width = 13cm]{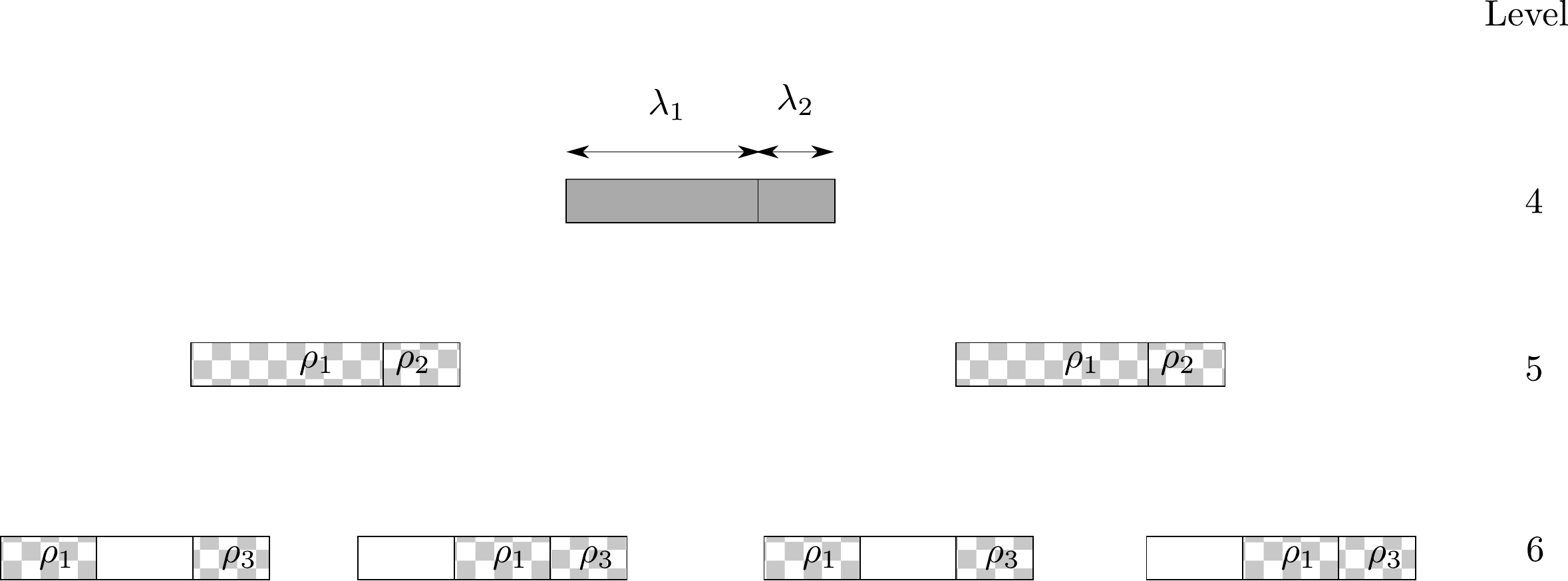}
	\caption{Detail of the floors where we use partial representations.}
	\label{fig:partial_detail}
\end{figure}

As explained in section \ref{subsec:representations}, most of the
elements we build at floors $5$ and $4$ are \textit{badly-formed} and do
not match the desired densities of $0$s, $1$s and $2$s. We only keep
the well-formed elements and lower the number of bits on which we
merge, so that the merged lists have again $L$ elements. In our case,
as the expected number of well-formed elements in level-$4$
lists is $2^{0.0116n}$, we merge on $2^{0.0055n}$ bits to compute the
level-$3$ lists (instead of $2^{0.0176n}$ bits.) Similarly, we merge
on $2^{0.0173n}$ bits to compute the level-$4$ lists because level-$5$
lists have $2^{0.0174n}$ well-formed elements. This is represented
in Figure \ref{fig:zoomed_tree}.

\begin{figure}[h!] \centering
	\includegraphics[width = 13cm]{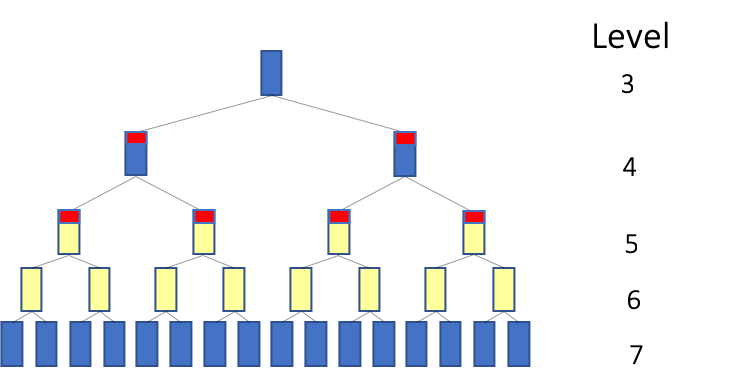}
	\caption{Detail of the bottom floors. Red elements are
          badly-formed elements.}
	\label{fig:zoomed_tree}
\end{figure}

Finally, level 7 is a left-right split with smaller lists (because of
the smoothing). The leaves have size $2^{0.01039n}$, so we merge on
$2^{0.0032n}$ bits to build the level-$6$ lists.

The numbers of well-formed elements per list are thus (from level $0$
to level $7$):
\[2^{0.0176n}, 2^{0.0176n}, 2^{0.0176n}, 2^{0.0176n}, 2^{0.0116n},
2^{0.0174n}, 2^{0.0176n}, 2^{0.01039n},\]
and the numbers of bits we merge on:
\[2^{0.0176n}, 2^{0.0176n}, 2^{0.0176n}, 2^{0.0055n}, 2^{0.0173n},
2^{0.0176n}, 2^{0.0032n}.\]

One can check that we merge on a total of $2^{0.0964n}$ bits, which is
exactly equal to $2^{\ell \log_2(3)}$, meaning that the level-$0$ list
is entirely composed of solutions of the Subset Sum problem.

One can also check that the Subset Sum problem has $2^{\ell+k-\ell
\log_2(3)} = 2^{0.6404n}$ solutions, that one solution has
$2^{0.4915n}$ representations (see appendix \ref{sec:ternaryReps}),
and that the merging constraints waste $2^{1.1143n}$ solution
representations. We are thus left with $2^{0.0176n}$ solutions, which
are exactly the solutions we get on the level-$0$ list.

\subsection{Summary of our Results}

We present here $2$ plots that illustrate the performance of our
different algorithms. What we show is that, in this parameter range,
the gain obtained by using representations is relatively small. This
is quite surprising because, in the binary case, representations are
very efficient. One explanation we have is that, in a regime where
there are naturally many solutions, Wagner's algorithm is very
efficient while the representation technique has difficulties in
finding solutions in amortized time $O(1)$. In Section \ref{sec:Hard},
we study the hardest instances, and show that representations turn out
to be more efficient.

\begin{figure}[!h!] \centering
	\includegraphics[width = 12cm]{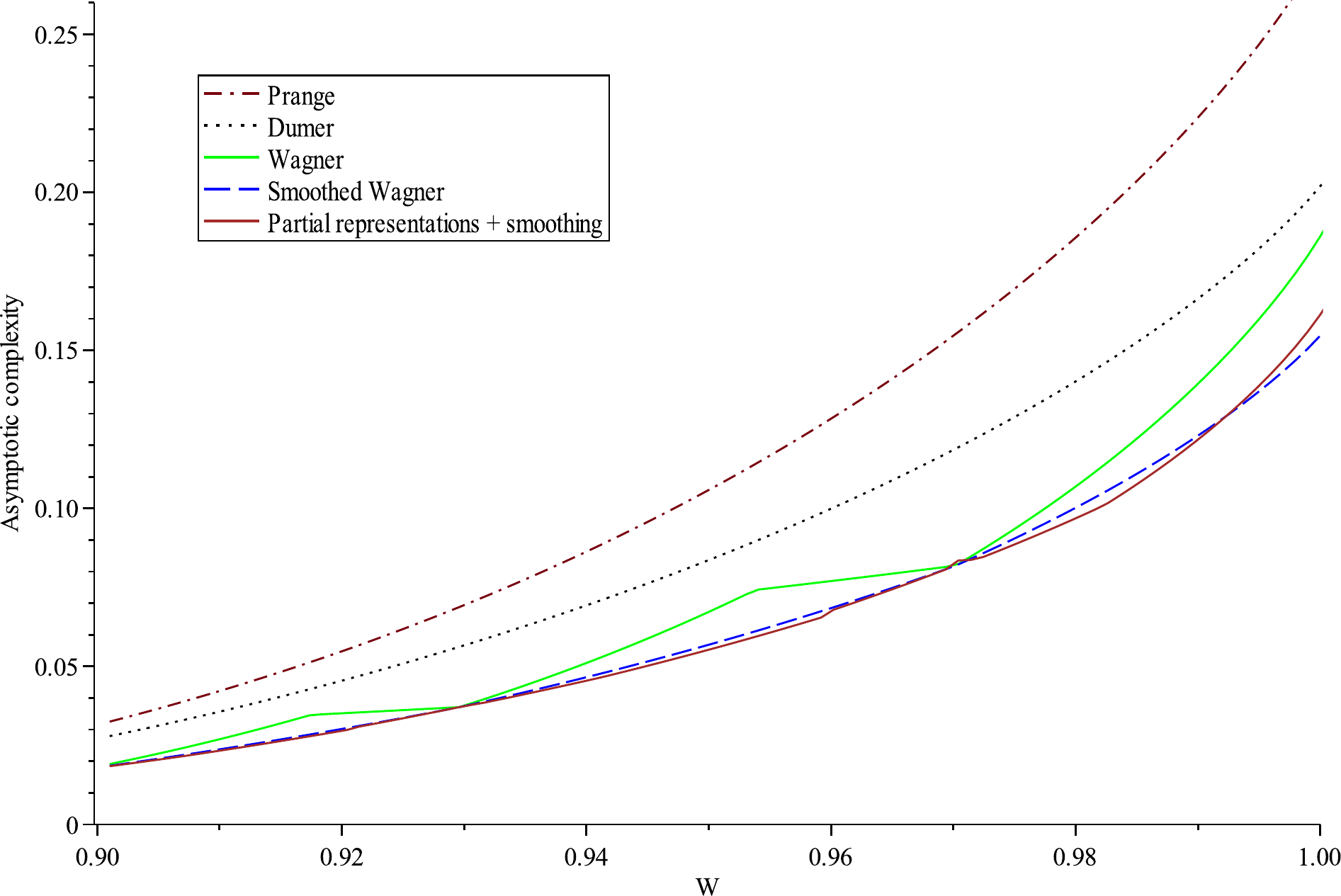}
	\caption{Comparison of the exponent complexities for $R = 0.5$}
	\label{fig:complexity_R5}
\end{figure}

\begin{figure}[!h!] \centering
	\includegraphics[width = 12cm]{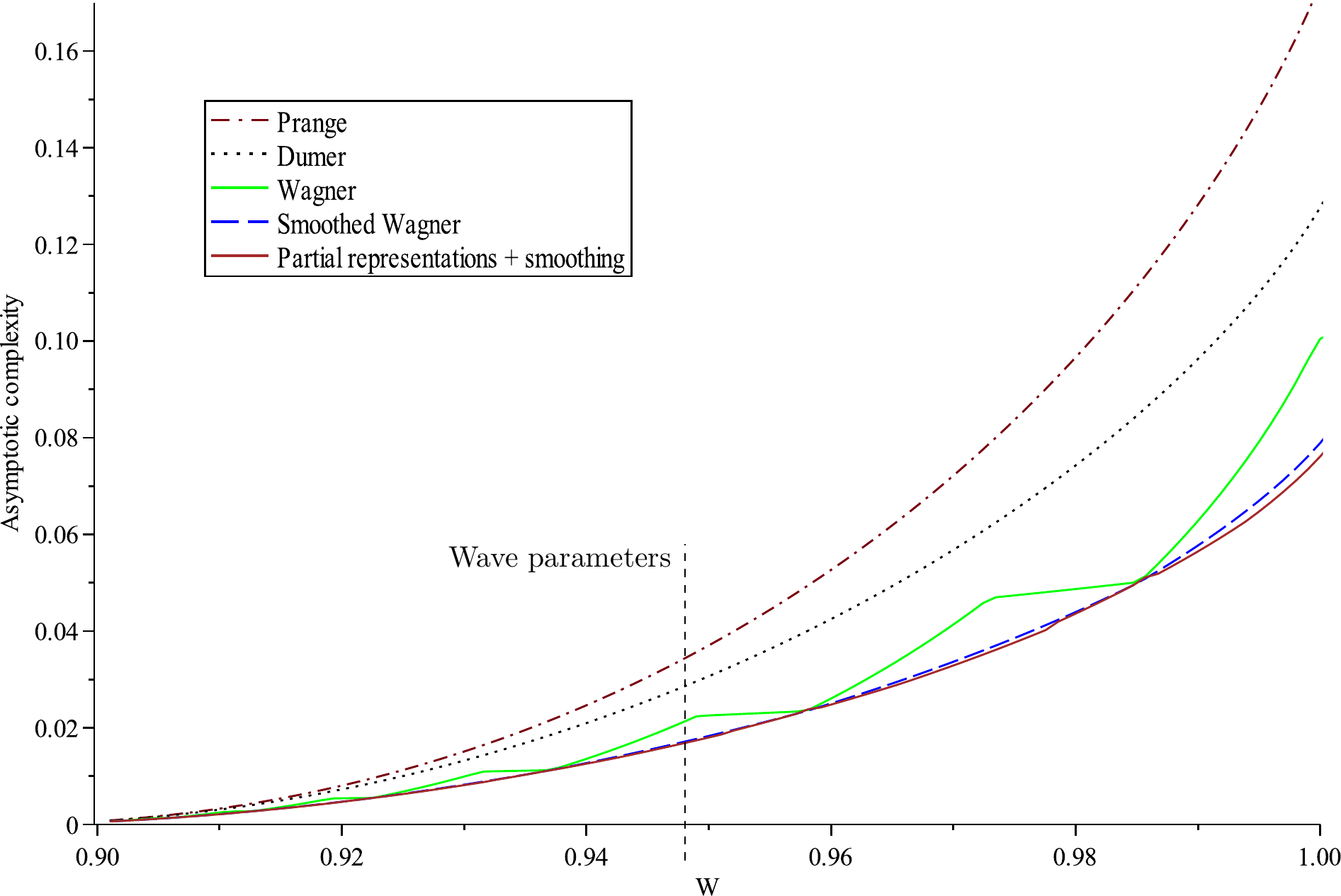}
	\caption{Comparison of the exponent complexities for $R = 0.676$}
	\label{fig:complexity_R676}
\end{figure}

%% file: 5-Wave.tex
\wave\ is a new code-based signature scheme proposed in
\cite{DST18}. It uses a \emph{hash-and-sign} approach and follows the
GPV paradigm \cite{GPV08} with the instantiation of a code-based
preimage sampleable family of functions.

Forging a signature in the Wave scheme amounts to solving the SD
problem. Roughly speaking, the public key is a specific pseudo-random
parity-check matrix $\Hm$ of size $(n-k)\times n$ and the signature of
a message $\mv$ is an error $\ev$ of weight $w$ such that
$\ev\transpose{\Hm} = h(\mv)$ with $h$ a hash function.  However,
instead of trying to forge a signature for one message of our choice,
a natural idea is to try to forge one message among a selected set of
messages. This context leads directly to a slight variation of the
classical SD problem. Instead of having one syndrome, there is a list
of possible syndromes and the goal is to decode one of them.  This
problem is known as the \emph{Decoding One Out of Many} (DOOM)
problem.

\begin{restatable}{problem}
	{problemDOOM}[Decoding One Out of Many - $\DOOM(n,z,q,R,W)$]
	\label{prob:DOOMH}
	
	\begin{tabular}{ll} 
		Instance: &
		\quad $\Hm\in\F_{q}^{(n-k)\times n}$ of full rank,\\ &
		\quad $\sv_1,\cdots,\sv_z\in\F_{q}^{n-k}$. \\ 
		Output: & \quad $\ev \in\F_{q}^n \textrm{ and } i \in \llbracket 1,z \rrbracket$ such that
		$\wt{\ev}=w$ and $\ev\transpose{\Hm} =\sv_i$,
	\end{tabular} $ \ $
	
	where $k \eqdef \lceil Rn\rceil$, $w \eqdef \lceil W n\rceil$
	and $\wt{\ev} \eqdef |\{i : \ev_i \neq0\}|$.
\end{restatable}

This problem was first considered in \cite{JJ02} and later analyzed
for the binary case $(q=2)$ in \cite{S11,DST17a}. These papers show
that one can solve the DOOM problem with an exponential speed-up
compared to the $\SD$ problem with equivalent parameters.

The difference induced by DOOM on the PGE+SS framework is that it
increases the search space. Namely, instead of searching a solution
$\ev$ of weight $w$ in the space
$\left\{ \ev \; : \; \ev\transpose{\Hm} = \sv \right\}$ we search in
$\cup_{i=1}^{z} \left\{ \ev \; : \; \ev\transpose{\Hm} = \sv_i
\right\}$. 

The idea to solve this problem with Wagner's approach is to take
$z \geq 3^{\ell/a}$ and replace the bottom-right list of the tree
$\Lc_{2^a}$ by a list containing all the syndromes. Hence, there are
only $2^{a}-1$ lists to generate from the search space. Therefore, the
constraint of Theorem \ref{Theorem:Wagner} becomes
\begin{equation*}
  3^{\ell/a} \leq 2^{(k+\ell)/(2^{a}-1)}.
\end{equation*}

For the practical parameters, we have $a = 6$ or $a=7$ so the change
from $2^a$ to $2^a-1$ has a negligible impact when we adapt the
representation technique to the DOOM problem.

The DOOM parameters stated in \cite{DST18} are derived from the
complexity of a key attack detailed in the \wave\ paper. Our result
stated in Section \ref{subsec:partial-rep-complexity} provides another
attack to consider. We computed the minimal parameters for the \wave\
scheme so that both attacks would have a time complexity of at least
$2^{128}$. They are stated in Table \ref{tab:wave-parameters} where
$n$ is the length of code used in Wave, $k$ its dimension and $w$ the
weight of the signature. These should be considered as the new
parameters to use for the \wave\ scheme.

\begin{table}[h!]  \centering
	\begin{tabular}{c c c}
          \toprule
          \quad $(n,k,w)$ \quad & \quad Public key size (in MB) \quad & \quad Signature length (in kB) \quad    \\
          \midrule
          (7236,4892,6862)&  2.27 &  1.434 \\ 
          \bottomrule
	\end{tabular}
	\vspace*{+0.3cm}
	\caption{New parameters of the \wave\ signature scheme for $128$ bits of security.}
        \label{tab:wave-parameters}
\end{table}

%% file: 6-hardest-instances.tex
In the previous sections, we tried to optimize our algorithms for the
regime of parameters used in the {\wave} signature scheme. The
corresponding Syndrome Decoding problem uses $R = 0.676$ and
$w \approx 0.948$. This corresponds to a regime where there are many
solutions to the problem and hence Wagner's algorithm with a large
number of floors was efficient. However, this is not the problem where
the problem is the hardest.

We will now look at the hardest instances of the ternary Syndrome
Decoding in large weight. As we already teased in the introduction,
ternary SD is much harder in large weights than in small weights. In
the two examples we considered, namely $R=0.5$ and $R=0.676$, the
problem was the hardest for $W = 1$. As we will see, there are some
lower rates for which the complexity of the Syndrome Decoding problem
is maximal for $W < 1$.

Consider an instance $(\Hm,\sv)$ of SD$(3,R,W)$ with
$W \ge \frac{2}{3}$. We have $\Hm\in\F_{3}^{(n-k)\times n}$ of full
rank and $\sv\in\F_{3}^{n-k}$. As in the binary case, the problem is
the hardest when it has a unique solution on average, if such a regime
exists.

Let $R_{\textup{max}} \eqdef \frac{\log_2(3) - 1}{\log_2(3)} \approx
0.36907$.  For $R \in [0,R_{\textup{max}}]$, we define
$W_{\textup{GV}}^{\textup{high}}(R)$ as the only value in $[2/3,1]$ such that
\[W_{\textup{GV}}^{\textup{high}}(R) + h_2(W_{\textup{GV}}^{\textup{high}}(R))
= (1-R)\log_2(3),\]
where $h_2(x) \eqdef -x\log_2(x) - (1-x)\log_2(1-x)$. 

The rate $R_{\textup{max}}$ was defined such that
$W_{\textup{GV}}^{high}(R_{\textup{max}}) = 1$, while this quantity is
not defined for $R > R_{\textup{max}}$. This is why Figures
\ref{fig:complexity_R5} and \ref{fig:complexity_R676} do not show a
high peak for $R = 0.5$ and $R=0.676$ but an increasing function up to
$W=1$. By Proposition \ref{propo:GVdist}, quantity
$W_{\textup{GV}}^{\textup{high}}(R)$ corresponds to the relative
weight for which we expect one solution to $SD$ with rate $R$ and
$q = 3$.
  
In order to study the above problem for hard high weight instances, we
compared the performance of $3$ standard algorithms: Prange's
algorithm, Dumer's algorithm and the BJMM algorithm.

We performed a case study and showed that, for all the above-mentioned
algorithms, the hardest case is reached for
$R = R_{\textup{max}} \approx 0.36907$ and $W = 1$. We obtain the
following results.

\begin{table}[h!]  \centering
	\begin{tabular}{c c c}
		\toprule
		\quad Algorithm \quad  & \quad $q=2$ \quad  & 
		\quad  $q = 3$ and $W > 0.5$ \quad  \\
		\midrule
		Prange & 0.121 ($R$ = 0.454) & 0.369 ($R$ = 0.369)  \\ 
		Dumer/Wagner & 0.116 ($R$ = 0.447) & 0.269 ($R$ = 0.369) \\ 
		BJMM/our algorithm & 0.102 ($R$ = 0.427)  & 0.247 ($R$ = 0.369) \\ 
		\bottomrule
	\end{tabular}
	\vspace*{+0.3cm}
	\caption{Best exponents with associated rates. \label{table:size2}}
\end{table}

In both the binary and the ternary case, we can see that Prange's
algorithm performs very poorly, but that Dumer's algorithm already
gives much better results and that BJMM's Subset Sum techniques, using
representations, increases the gain. The analysis of Prange and Dumer
for $q=3$ is quite straightforward and follows closely the binary
case. For BJMM (\textit{i.e.} Wagner's algorithm with
representations), the exponent $0.247$ comes from a $2$-levels Wagner
tree that includes $1$ layer of representations. We tried using a larger
Wagner trees but this did not give any improvement.

The ternary SD appears significantly harder than its binary
counterpart.  This was expected to some extent because in the ternary
case, the input matrices have elements in $\F_3$ and not $\F_2$, which
means that matrices of the same dimension contain more
information.

In order to confirm this idea, we define the following metric: what is
the smallest input size for which the algorithms need at least
$2^{128}$ operations to decode?  We use the value $128$, as $128$
security bits is a cryptographic standard. The input matrix
$\Hm \in \F_3^{n(1-R) \times n}$ is represented in systematic
form. This means that we write
 \[ \Hm =
\begin{pmatrix} \un_{n(1-R)} & \Hm'
\end{pmatrix}.
\] 
The only relevant part that needs to be specified is $\Hm'$. This
requires $R(1-R)n^2\log_2(q)$ bits.  We show that, even in this
metric, the ternary syndrome decoding problem is much harder,
\textit{i.e.} requires $2^{128}$ operations to decode inputs of much
smaller sizes. Our results are summarized in the table below.
 
\begin{table}[h!]  \centering
	\begin{tabular}{c c c}
		\toprule
		\quad Algorithm \quad  & \quad $q=2$ \quad  & 
		\quad  $q = 3$ and $W > 0.5$ \quad  \\
		\midrule
		Prange & 275 ($R$ = 0.384) & 44 ($R$ = 0.369) \\ 
		Dumer/Wagner & 295 ($R$ = 0.369) & 83 ($R$ = 0.369) \\ 
		BJMM/our algorithm & 374 ($R$ = 0.326) & 99 ($R$ = 0.369) \\ 
		\bottomrule
	\end{tabular}
	\vspace*{+0.3cm}
	\caption{Minimum input sizes (in Kbits) for a time complexity of $2^{128}$. \label{table:size3}}
\end{table}
 
Notice that in this metric, in the binary case, it is worth reducing
the rate $R$, as this reduces the input size. But in the ternary case,
we do not observe this behavior, which shows that the problem quickly
becomes simple, as $R$ decreases.

The work we present here is very preliminary but opens many new
perspectives. It seems there are many cases in code-based
cryptography, from encryption schemes to signatures, where this
problem could replace the binary Syndrome Decoding problem to get
smaller key sizes.

%% file: 7-conclusion.tex
In this work, we stressed a strong difference between the cases $q=2$
and $q \geq 3$ of the Syndrome Decoding problem. Namely, the symmetry
between the small weight and the large weight cases, which occurs in
the binary case, is broken for larger values of $q$. The large weight
case of the general Syndrome Decoding problem had never been studied
before. We proposed two algorithms to solve the Syndrome Decoding
problem in this new regime in the context of the \textit{Partial
Gaussian Elimination and Subset Sum} framework. Our first algorithm
uses a $q$-ary version of Wagner's approach to solve the underlying
Subset Sum problem. We proposed a second algorithm making use of
representations as in the BJMM approach. We studied both algorithms
and proposed a first application for cryptographic purposes, namely
for the \wave\ signature scheme. Considering our complexity analysis,
we proposed new parameters for this scheme. Furthermore, we showed
that the worst case complexity of Syndrome Decoding in large weight is
higher than in small weight. This implies that it should be possible
to develop new code-based cryptographic schemes using this regime of
parameters that reach the same security level with smaller key size.

%% file: appendix.tex
In this section, we explain how to compute the number of
representations as well as the number of badly-formed vectors when
using ternary representations.

\subsection{Notations}

The notation $\bin{n}{k_1, \dots, k_i}$ will denote the multinomial
coefficient $\dfrac{n!}{k_1! \dots k_i!}$, assuming that $n = k_1 +
\dots + k_i$.

Let us denote $g : (n, k_1, k_2) \rightarrow n\log_2(n) - k_1\log_2(k_1)
- k_2\log_2(k_2) - (n-k_1-k_2)\log_2(n-k_1-k_2)$. We have:
\[\bin{n}{k_1, k_2, n-k_1-k_2} = \OOt{2^{g(n, k_1, k_2)}}.\]

The function $g$ satisfies :
\begin{itemize}
\item $g(n, k_1, k_2) = g(n, k_2, k_1)$,
\item $g(n, k_1, k_2) = g(n, k_1, n-k_1-k_2)$,
\item $g(\lambda n, \lambda k_1, \lambda k_2) = \lambda g(n, k_1,
k_2)$,
\item $g(n, k_1, k_2) =
n h_2\left(\frac{k_1+k_2}{n}\right)+(k_1+k_2) h_2\left(\frac{k_1}{k_1+k_2}\right)$,
where $h_2$ stands for the binary entropy.
\end{itemize}

We denote by $T(n, \alpha, \beta)$ the set of all vectors of length $n$
composed with $\alpha n$ $1$s, $\beta n$ $2$s and $(1-\alpha-\beta)n$
$0$s. There exist $\bin{n}{\alpha n$, $\beta n$,
  $(1-\alpha-\beta)n} = \OOt{2^{ng(1, \alpha, \beta)}}$ such vectors.

\subsection{Main Result}
The goal of this section is to prove the following result.

\begin{proposition}\label{Prop:Reps!} For $\mathbf{b} \in T(n,
\alpha_0, \beta_0)$, the number of way one can decompose $\mathbf{b}$
as the sum of two vectors from $T(n, \alpha_1, \beta_1)$ is given by:
\[\OOt{ 2^{n(g(1-\alpha_0-\beta_0, \overline{x}_{12}, \overline{x}_{12}) + g(\alpha_0, \overline{x}_{01}, \overline{x}_{01}) + g(\beta_0, \overline{x}_{02}, \overline{x}_{02}))} },\]
where
\[\begin{array}{cl}
\overline{x}_{01} =& \dfrac{2\alpha_0 + \beta_0 - \alpha_1 -
2\beta_1}{3} +z\\ \overline{x}_{02} =& \dfrac{\alpha_0 + 2\beta_0 -
2\alpha_1 - \beta_1}{3} +z\\ \overline{x}_{12} =& z\\
\end{array}\]
 and $z$ is the real root of
\[\dfrac{
( 2\alpha_0 + \beta_0 - \alpha_1 - 2\beta_1 +3z ) ( \alpha_0 +
2\beta_0 - 2\alpha_1 - \beta_1 +3z ) z} {( 1-\alpha_0-\beta_0 - 2z ) (
-2\alpha_0 - \beta_0 +4 \alpha_1 +2\beta_1 - 6z ) ( -\alpha_0 -
2\beta_0 +2 \alpha_1 +4\beta_1 - 6z )} = 1.\]
\label{Nrep}
\end{proposition}

\subsection{A Simple Example}

Let us consider a very simple case: we want to decompose a balanced
vector of size $n$ (the number of $0$s, $1$s and $2$s is $n/3$) in two
balanced vectors (\textit{i.e.}
$\alpha_0 = \beta_1 = \alpha_1 = \beta_1 = 1/3$). There are several
ways to achieve this. One solution is that each $0$ is obtain by
$0 + 0$, each $1$ by $2+2$, and each $2$ by $1+1$. There is exactly
only one way to build the vector in this way. Another possibility is
that each case ($0+0$, $1+2$, $2+1$, $1+0$, $2+2$, $0+1$, $2+0$, $0+2$
and $1+1$) happens $n/9$ times. This is the scenario admitting the
maximal number of decompositions: $\OOt{3^n}$. There are many more
possibilities.

The number of representations is the sum of all decompositions for all
the possible scenarios. There are only a polynomial number of
different scenarios. The total number of representations (which is
what we want to determine) is determined, up to a polynomial factor,
by the scenario which gives the maximal number of decompositions. In
this case, there are $\OOt{3^n}$ representations.

Let us check that this is the result givent by Proposition
\ref{Nrep}. Indeed, in this case, $z$ must
satisfy the equation
\[\frac{(3z)(3z)z}{(1/3-2z)(1-6z)(1-6z)} = 1, \text{ or equivalently }
27z^3 = (1-6z)^3.\]

The real root of this equation is $1/9$. Thus we obtain
$\overline{x}_{01} = \overline{x}_{02} = \overline{x}_{12} =
1/9$. Finally, the number of representations is
\[
\OOt{2^{3ng(1/3, 1/9, 1/9)}} = \OOt{2^{3n \times (\log_2(3)/3)}} = \OOt{3^n}.
\]

\subsection{Typical Case}

In general, we have a vector $\bv \in T(n, \alpha_0, \beta_0)$. We want
to decompose it into two vectors of $T(n, \alpha_1, \beta_1)$. Let us
call $x_{00}, \dots, x_{22}$ the density of the nine cases ($0+0$,
$0+1$, $\dots$, $2+2$) as shown in the following table :
\[
\begin{array}{p{3mm}|ccc}
 & 0 & 1 & 2 \\ \hline 
0 & x_{00} & x_{01} & x_{02}\\ 
1 & x_{10} & x_{11} & x_{12}\\
2 & x_{20} & x_{21} & x_{22}
\end{array}
\]

We denote by $\mathcal{A}$ the set of possible such tuples
$x_{00}, \dots, x_{22}$.

Given the target vector $\mathbf{b}$, there are
$\bin{n(1-\alpha_0-\beta_0)}{nx_{00}, nx_{12},
nx_{21}}\bin{n\alpha_0}{nx_{01}, nx_{10},
nx_{22}}\bin{n\beta_0}{nx_{02}, nx_{11}, nx_{20}}$ ways of decomposing
this $\mathbf{b}$ according to $(x_{00}, \dots, x_{22})$. Indeed, a
$0$ in $\mathbf{b}$ can be decomposed as $0+0$ (this happens $n
x_{00}$ times), $1+2$ ($n x_{12}$ times) or $2+1$ ($n x_{21}$
times). As the number of $0s$ in $\mathbf{b}$ is
$n(1-\alpha_0-\beta_0)$, there are
$\bin{n(1-\alpha_0-\beta_0)}{nx_{00}, nx_{12}, nx_{21}}$ ways to
choose the decomposition of each $0$ of $\mathbf{b}$. The choices of
the decompositions of the $1$s and the $2$s give the other two
factors.

For given $\alpha_0, \beta_0, \alpha_1$ and $\beta_1$, the number of
possible decompositions is
\[
\sum\limits_{(x_{00}, \dots, x_{22}) \in
  \mathcal{A}}\bin{n(1-\alpha_0-\beta_0)}{nx_{00}, nx_{12},
  nx_{21}}\bin{n\alpha_0}{nx_{01}, nx_{10},
  nx_{22}}\bin{n\beta_0}{nx_{02}, nx_{11}, nx_{20}}.
\]

Up to a polynomial factorm this is equal to
\[
\sum\limits_{(x_{00}, \dots, x_{22}) \in
  \mathcal{A}}2^{n(g(1-\alpha_0-\beta_0, x_{21}, x_{12}) + g(\alpha_0,
  x_{01}, x_{10}) + g(\beta_0, x_{02}, x_{20}))}.
\]

The largest term (or eventually one of the largest terms) of this sum is
\[
2^{n(g(1-\alpha_0-\beta_0, \overline{x}_{21}, \overline{x}_{12}) + g(\alpha_0, \overline{x}_{01}, \overline{x}_{10}) + g(\beta_0, \overline{x}_{02}, \overline{x}_{20}))},
\]
where $(\overline{x}_{00}, \dots, \overline{x}_{22})$ is called the
\textit{typical case}.

We are interested in this typical case because it gathers a polynomial
fraction of all the possible decompositions. The asymptotic exponent
of the total number of representations is then simply given by the
exponent of the typical case.

\subsection{Computation of the Typical Case}

Given $\alpha_0, \beta_0, \alpha_1$ and $\beta_1$, the following
constraints exist on $x_{00}, \dots, x_{22}$.
\[
\begin{array}{ccl}
x_{00}+x_{01}+x_{02} &=& 1-\alpha_1-\beta_1\\ x_{10}+x_{11}+x_{12} &=&
\alpha_1\\ x_{20}+x_{21}+x_{22} &=& \beta_1\\ \\ x_{00}+x_{10}+x_{20}
&=& 1-\alpha_1-\beta_1\\ x_{01}+x_{11}+x_{21} &=& \alpha_1\\
x_{02}+x_{12}+x_{22} &=& \beta_1\\ \\ x_{00}+x_{12}+x_{21} &=&
1-\alpha_0-\beta_0\\ x_{01}+x_{10}+x_{22} &=& \alpha_0\\
x_{02}+x_{11}+x_{20} &=& \beta_0\\
\end{array}
\]

However, these equations are not independent. Each of the three sets
of three equations implies $x_{00} + \dots + x_{22}=1$. We are
actually left with two degrees of freedom, and any solution can be
written as
\[
\begin{array}{ccccccc}
x_{00} &=& 1-\alpha_0-\beta_0& - &2z\\ 
x_{01} &=& \dfrac{2\alpha_0 + \beta_0 - \alpha_1 - 2\beta_1}{3} & +&z & +&w\\
x_{02} &=& \dfrac{\alpha_0 + 2\beta_0 - 2\alpha_1 - \beta_1}{3} & +&z & -&w\\
x_{10} &=& \dfrac{2\alpha_0 + \beta_0 - \alpha_1 - 2\beta_1}{3} & +&z & -&w\\
x_{11} &=& \dfrac{-2\alpha_0 - \beta_0 +4 \alpha_1 +2\beta_1}{3}&-& 2z\\
x_{12} &=& 0& +&z & +&w\\
x_{20} &=& \dfrac{\alpha_0 + 2\beta_0 -2\alpha_1 - \beta_1}{3} & +&z & +&w\\ 
x_{21} &=& 0& +&z & -&w\\ 
x_{22} &=&\dfrac{-\alpha_0 - 2\beta_0 +2 \alpha_1 +4\beta_1}{3}& -& 2z.
\end{array}
\]

Thus, $\mathcal{A} = \{(x_{00}(w, z), \dots, x_{22}(w, z)) \;|\;
\forall (i, j), x_{ij} \geqslant 0\}$.

\subsubsection*{Determining $w$.}

In a first step, we will show that the typical case must be symmetric
(\textit{i.e.} $\overline{x}_{01}=\overline{x}_{10}$,
$\overline{x}_{02}=\overline{x}_{20}$ and
$\overline{x}_{12}=\overline{x}_{21}$), which means that $w$ must be $0$.  To do so, we consider a pair
$(w, z)$ such that the corresponding $(x_{00}, \dots, x_{22})$ is in
$\mathcal{A}$, and we call $(\tilde{x}_{00}, \dots , \tilde{x}_{22})$
the solution with the same $z$ but $0$ instead of $w$.

As $(x_{00}, \dots, x_{22})$ is in $\mathcal{A}$, all $x_{ij}$ are
positive or zero. This implies that all $\tilde{x}_{ij}$ are positive
or zero. For example, for $\tilde{x}_{01}$ we have:
\[
0 \leqslant \min(x_{01}, x_{10}) = \tilde{x}_{01} - \abs(w) \leqslant \tilde{x}_{01}.
\]

Therefore, $(\tilde{x}_{00}, \dots, \tilde{x}_{22})$ is in $\mathcal{A}$ and
we obtain
\begin{equation}
\label{eq:nrep}
\dfrac{\Nrep(\mathbf{\tilde{x}})}{\Nrep(\mathbf{x})} =
\dfrac{2^{n(g(1-\alpha_0-\beta_0, \tilde{x}_{21}, \tilde{x}_{12}) +
    g(\alpha_0, \tilde{x}_{01}, \tilde{x}_{10}) + g(\beta_0,
    \tilde{x}_{02}, \tilde{x}_{20}))}}{2^{n(g(1-\alpha_0-\beta_0,
    x_{21}, x_{12}) + g(\alpha_0, x_{01}, x_{10}) + g(\beta_0, x_{02},
    x_{20}))}}.
\end{equation}

But we have the following equality.
\[
\begin{array}{ccl}
g(1-\alpha_0-\beta_0, x_{21}, x_{12}) &=& g(1-\alpha_0-\beta_0,
\tilde{x}_{12}+w, \tilde{x}_{12}-w)\\ &=& g(1-\alpha_0-\beta_0,
\tilde{x}_{12}, \tilde{x}_{12}) +
2\tilde{x}_{21}\left(h(1/2+w/\tilde{x}_{12}) - h(1/2) \right).
\end{array}\]

Similarly, we obtain two other formulas.
\[g(\alpha_0, x_{01}, x_{10}) = g(\alpha_0, \tilde{x}_{01}, \tilde{x}_{10}) + 2\tilde{x}_{01}\left(h(1/2+w/\tilde{x}_{01}) - h(1/2) \right),\]
\[g(\beta_0, x_{02}, x_{20}) = g(\beta_0, \tilde{x}_{02}, \tilde{x}_{20}) + 2\tilde{x}_{02}\left(h(1/2+w/\tilde{x}_{02}) - h(1/2) \right).\]

Therefore, we can reduce Equation \ref{eq:nrep} to 
\[\dfrac{\Nrep(\mathbf{\tilde{x}})}{\Nrep(\mathbf{x})} =
2^{2n\left(\tilde{x}_{01}\left(1-h\left(\frac{1}{2}+\frac{w}{\tilde{x}_{01}}\right)\right)
    +
    \tilde{x}_{02}\left(1-h\left(\frac{1}{2}+\frac{w}{\tilde{x}_{02}}\right)\right)
    +
    \tilde{x}_{12}\left(1-h\left(\frac{1}{2}+\frac{w}{\tilde{x}_{12}}\right)\right)\right)}.\]

So $\Nrep(\mathbf{x}) \leqslant \Nrep(\mathbf{\tilde{x}})$ and these
two quantities are equal if and only $w=0$, \textit{i.e.}
$\mathbf{x} = \mathbf{\tilde{x}}$.

\subsubsection*{Determining $z$.}

To get the typical case, we now have to find the value of $z$ that
maximises the expression
\[
g(1-\alpha_0-\beta_0, x_{21}, x_{12}) + g(\alpha_0, x_{01}, x_{10}) +
g(\beta_0, x_{02}, x_{20}).
\]

Notice that this expression is equivalent, up to an additive constant,
to $-\sum\limits_{i, j} x_{ij}\log_2(x_{ij})$.

This function is concave and thus admit a single maximum. The
differentiation of this function with respect to $z$ gives
\[
-2\log_2 \left( \dfrac{ \left( \dfrac{2\alpha_0 + \beta_0 - \alpha_1 -
        2\beta_1}{3} +z \right) \left( \dfrac{\alpha_0 + 2\beta_0 -
        2\alpha_1 - \beta_1}{3} +z \right) z} {\left(
      1-\alpha_0-\beta_0 - 2z \right) \left( \dfrac{-2\alpha_0 -
        \beta_0 +4 \alpha_1 +2\beta_1}{3} - 2z \right) \left(
      \dfrac{-\alpha_0 - 2\beta_0 +2 \alpha_1 +4\beta_1}{3} - 2z
    \right)} \right),
\]
which is equal to zero if and only if
\[
\dfrac{
( 2\alpha_0 + \beta_0 - \alpha_1 - 2\beta_1 +3z ) ( \alpha_0 +
2\beta_0 - 2\alpha_1 - \beta_1 +3z ) z} {( 1-\alpha_0-\beta_0 - 2z ) (
-2\alpha_0 - \beta_0 +4 \alpha_1 +2\beta_1 - 6z ) ( -\alpha_0 -
2\beta_0 +2 \alpha_1 +4\beta_1 - 6z )} = 1.
\]

This explains why $z$ is the root of a polynomial of degree $3$.

\subsection{Number of Representations and Badly-formed Elements}

There are $\OOt{2^{ng(1, \alpha_0, \beta_0)}}$ vectors in
$T(n, \alpha_0, \beta_0)$. For each of these vectors, there are by
definition $\Nrep(\alpha_0, \beta_0, \alpha_1, \beta_1)$ ways of
decomposing it as the sum of two vectors of $T(n, \alpha_1,
\beta_1)$.
Moreover, the number of vectors in $T(n, \alpha_1, \beta_1)$ is
$\OOt{2^{ng(1, \alpha_1, \beta_1)}}$. There are then
$\OOt{2^{2ng(1, \alpha_1, \beta_1)}}$ pairs of vectors of
$T(n, \alpha_1, \beta_1)$, but only
$\OOt{\Nrep(\alpha_0, \beta_0, \alpha_1, \beta_1)2^{ng(1, \alpha_0,
    \beta_0)}}$
of these pairs give a valid representation of a vector of
$T(n, \alpha_0, \beta_0)$. All the other pairs give badly-formed
elements. Thus, when we merge two $L$-sized lists of elements of
$T(n, \alpha_1, \beta_1)$ on $L$ bits, we obtain
$\OOt{L \Nrep(\alpha_0, \beta_0, \alpha_1, \beta_1)2^{ng(1, \alpha_0,
    \beta_0) - 2ng(1, \alpha_1, \beta_1)}}$
vectors of $T(n, \alpha_0, \beta_0)$, the remaining consisting on
badly-formed vectors.